\documentclass[twoside,10pt]{article} 
\usepackage{amsfonts}
\usepackage{amsmath}
\usepackage{amssymb}
\usepackage{epsfig}
\usepackage{graphicx}
\usepackage{latexsym}
\usepackage{subfigure}   
\usepackage{times}
\usepackage{hyperref}

\newcommand{\app}{\textsc{Hamlet}}
\newcommand{\greta}{\textsc{Greta}}
\newcommand{\sharon}{\textsc{Sharon}}
\newcommand{\mcep}{\textsc{MCEP}}

\newcommand{\within}{\textsf{\footnotesize WITHIN}}
\newcommand{\slide}{\textsf{\footnotesize SLIDE}}

\newcommand{\seq}{\textsf{\footnotesize SEQ}}

\newcommand{\mycount}{\textsf{\footnotesize COUNT}}
\newcommand{\mymin}{\textsf{\footnotesize MIN}}
\newcommand{\mymax}{\textsf{\footnotesize MAX}}
\newcommand{\mysum}{\textsf{\footnotesize SUM}}
\newcommand{\myavg}{\textsf{\footnotesize AVG}}
\newcommand{\mynot}{{\small \textsf{NOT}}}

\makeatletter
\providecommand*{\cupdot}{%
  \mathbin{%
    \mathpalette\@cupdot{}%
  }%
}
\newcommand*{\@cupdot}[2]{%
  \ooalign{%
    $\m@th#1\cup$\cr
    \hidewidth$\m@th#1\cdot$\hidewidth
  }%
}

\usepackage[ruled]{algorithm}
\PassOptionsToPackage{noend}{algpseudocode}
\usepackage{algpseudocode}

\renewcommand{\algorithmiccomment}[1]{\bgroup\hfill//~#1\egroup}

\algnewcommand\algorithmicswitch{\textbf{switch}}
\algnewcommand\algorithmiccase{\textbf{case}}
\algnewcommand\algorithmicassert{\texttt{assert}}
\algnewcommand\Assert[1]{\State \algorithmicassert(#1)}%
\algdef{SE}[SWITCH]{Switch}{EndSwitch}[1]{\algorithmicswitch\ #1\ \algorithmicdo}{\algorithmicend\ \algorithmicswitch}%
\algdef{SE}[CASE]{Case}{EndCase}[1]{\algorithmiccase\ #1}{\algorithmicend\ \algorithmiccase}%
\algtext*{EndSwitch}%
\algtext*{EndCase}%

\usepackage[font=small,labelfont=bf,tableposition=top]{caption}

\usepackage{listings}
\renewcommand{\algorithmiccomment}[1]{/* #1 */}

\usepackage{amsthm}
\usepackage{setspace,xspace}
\graphicspath{{figures/}}

\newcommand{\nop}[1]{}

\newcommand{\rem}[1]{\marginpar{\flushleft{#1}}}
\renewcommand{\rem}[1]{} 

\renewcommand{\algorithmiccomment}[1]{/* #1 */}

\graphicspath{{figures/}}
\newtheorem{definition}{Definition}
\newtheorem{example}{Example}
\newtheorem{theorem}{Theorem}[section]

\usepackage{multirow}
\newcommand{\eat}[1] {}

\usepackage{fancyhdr}

\fancypagestyle{plain}{%
  \fancyhead{}
  
}

 \newlength{\hoehe}
 \newlength{\breite}
\title{\fontsize{15}{15}\selectfont To Share, or not to Share Online Event Trend Aggregation\\ Over Bursty Event Streams\\\vspace*{1cm}
\large Technical Report\\
\vspace*{1cm}}
\author{\large Olga Poppe,$^1$ Chuan Lei,$^2$ Lei Ma,$^3$ Allison Rozet,$^4$ and Elke A. Rundensteiner$^3$}
\date{\Large 
January, 2021\\
\vspace*{1cm}
\large $^1$Microsoft Gray Systems Lab, One Microsoft Way, Redmond, WA 98052\\
$^2$IBM Research, Almaden, 650 Harry Rd, San Jose, CA 95120\\
$^3$Worcester Polytechnic Institute, Worcester, MA 01609\\
$^4$MathWorks, 1 Apple Hill Dr, Natick, MA 01760\\
\vspace*{0.2cm}
olpoppe@microsoft.com, chuan.lei@ibm.com, lma5@wpi.edu, arozet@mathworks.com, rundenst@wpi.edu\\
\vspace*{7cm}
}

\begin{document}
\maketitle

\begin{spacing}{0.8}
{\footnotesize \noindent \textbf{Copyright} \copyright{} 2021 by
authors. Permission to make digital or hard copies of all or
part of this work for personal use is granted without fee provided
that copies bear this notice and the full citation on the first
page. To copy otherwise, to republish, to post on servers or to
redistribute to lists, requires prior specific permission. }
\end{spacing}

\clearpage
\pagestyle{fancy}

\clearpage
\tableofcontents

\pagenumbering{arabic}
\setcounter{page}{1}

%
%

\newpage
\begin{abstract}

Complex event processing (CEP) systems continuously evaluate large workloads of pattern queries under tight time constraints. Event trend aggregation queries with Kleene patterns are commonly used to retrieve summarized insights about the recent trends in event streams. State-of-art methods are limited either due to repetitive computations or unnecessary trend construction. Existing shared approaches are guided by statically selected and hence rigid sharing plans that are often sub-optimal under stream fluctuations. In this work, we propose a novel framework \app\ that is the first to overcome these limitations. \app\ introduces two key innovations. First, \app\ adaptively decides whether to share or not to share computations depending on the current stream properties at run time to harvest the maximum sharing benefit. Second, \app\ is equipped with a highly efficient shared trend aggregation strategy that avoids trend construction. Our experimental study on both real and synthetic data sets demonstrates that \app\ consistently reduces query latency by up to five orders of magnitude compared to the state-of-the-art approaches.

\end{abstract}

\section{Introduction}
\label{sec:introduction}

Sensor networks, web applications, and smart devices produce high velocity event streams. Industries use Complex Event Processing (CEP) technologies to extract insights from these streams using Kleene queries~\cite{ADGI08,Giatrakos2020,ZDI14},
i.e., queries with Kleene plus ``+'' operator that matches event sequences of any length, a.k.a. \textit{event trends}~\cite{PLAR17}. Since these trends can be arbitrarily long and complex and there also tends to be a large number of them, they are typically aggregated to derive summarized insights~\cite{QCRR14}. CEP systems must thus process large workloads of these event trend aggregation queries over high-velocity streams in near real-time.

\begin{example}
Complex event trend aggregation queries are  used in Uber and DoorDash for price computation, forecasting, scheduling, and routing~\cite{uber-athenax}. With hundreds of users per district, thousands of transactions, and millions of districts nationwide, real-time event analytics has become a challenging task.

In Figure~\ref{fig:queries}, the query workload  computes various trip statistics such as the number, total duration, and average speed of trips per district. Each event in the stream is of a particular event type, e.g., \textit{Request}, \textit{Pickup}, \textit{Dropoff}. Each event is associated with attributes such as a time stamp, district, speed, driver, and rider identifiers.

Query $q_1$ focuses on trips in which the driver drove to a pickup location but did not pickup a rider within 30 minutes since the request. Each trip matched by $q_1$ corresponds to a sequence of one ride \textit{Request} event, followed by one or more \textit{Travel} events (expressed by the Kleene plus operator ``+''), and not followed by a \textit{Pickup} event. All events in a trip must have the same driver and rider identifiers as required by the predicate [driver, rider].
Query $q_2$ targets \textit{Pool} riders who were dropped off at their destination. 
Query $q_3$ tracks riders who cancel their accepted requests while the drivers were stuck in slow-moving traffic.
All three queries contain the expensive Kleene sub-pattern $T+$ that matches arbitrarily long event trends. Thus, one may conclude that sharing $T+$ always leads to computational savings. However, a closer look reveals that the actual sharing benefit depends on the current stream characteristics. Indeed, trips are affected by many factors, from time and location to specific incidents, as the event stream fluctuates. 
%
%
%
\label{ex:motivating}
\end{example}

\vspace*{-2mm}
\begin{figure}[!htb]
\centering
\includegraphics[width=0.5\columnwidth]{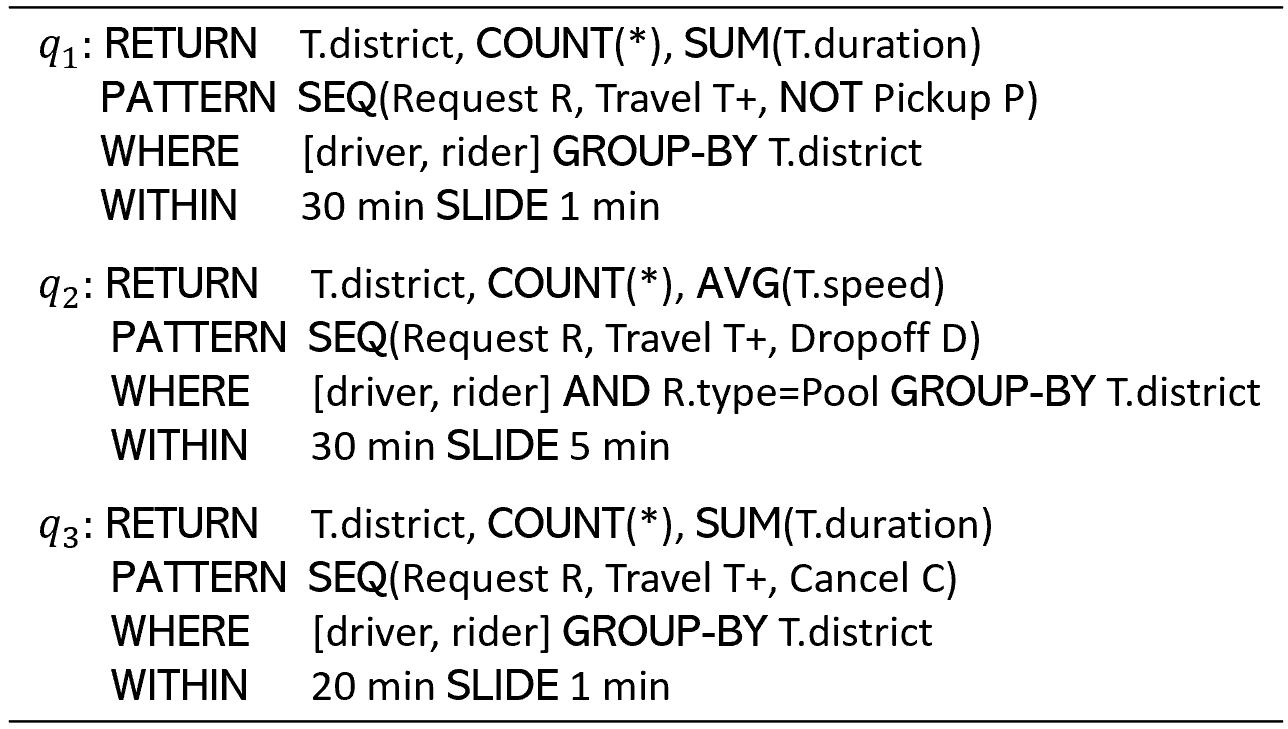}
\caption{Event trend aggregation queries}
\label{fig:queries}
\end{figure}
\vspace*{-2mm}


\textbf{Challenges}. To enable shared execution of trend aggregation queries, we must tackle the following open challenges. 


\textit{Exponential complexity versus real-time response}.
Construction of event trends matched by a Kleene query has exponential time complexity in the number of matched events~\cite{PLAR17,ZDI14}. To achieve real-time responsiveness, shared execution of trend aggregation queries should thus adopt online strategies that compute trend aggregates on-the-fly while avoiding this expensive trend construction~\cite{PLRM18,PLRM19}. However, shared execution applied to such online trend aggregation incurs additional challenges not encountered by the shared construction of traditional queries~\cite{KS19}. In particular, we must avoid constructing these trends, while  capturing  critical connections among shared sub-trends compactly to validate predicates of each query. 
For example, query $q_1$ in Figure~\ref{fig:queries} may match all events of type \textit{Travel}, while queries $q_2$ and $q_3$ may only match some of them due to their predicates. Consequently, different trends will be matched by these queries.
%
On first sight it appears that result validation requires the construction of all trends per query, which would defeat the goal of online aggregation. To address this dilemma, we must  develop  a correct yet efficient shared online trend aggregation strategy.


\textit{Benefit versus overhead of sharing}.
One may assume that the more sub-patterns are shared, the greater the performance improvement will be. However, this assumption does not always hold due to the overhead caused by maintaining intermediate aggregates of sub-patterns to ensure correctness of results. The computational overhead incurred by  shared query execution does not always justify the savings achievable compared to  baseline non-shared execution.
For example, sharing query $q_1$ with the other two queries in Figure~\ref{fig:queries} will not be beneficial if there are only few \textit{Pool} requests and the travel speed is above 10 mph.
%
Hence, we need to devise a lightweight benefit model that accurately estimates the benefit of shared execution of multiple trend aggregation  at runtime.

\textit{Bursty event streams versus light-weight sharing decisions}.
The actual sharing benefit can vary over time due to the nature of bursty event streams. Even with an efficient shared execution strategy and an accurate sharing benefit model, a static sharing solution may not always lead to computational savings. Worse yet, 
in some cases,  a static sharing decision may do more harm than good. 
Due to different predicates and windows of queries in Figure~\ref{fig:queries}, one may decide at compile time that these queries should not be shared. However, a large burst of \textit{Pool} requests may arrive and the traffic may be moving slowly (i.e., speed below 10 mph) in rush hour, making sharing of these queries beneficial.
For this, a dynamic sharing optimizer, capable of adapting to changing arrival rates, data distribution, and other cost factors, must be designed. Its runtime sharing decisions must be light-weight to ensure real-time responsiveness.


\begin{table}[!htb]
    \centering
    \begin{tabular}{|l||c|c|c|}
    \hline
        \multirow{2}{*}{\textbf{Approach}}
        & \textbf{Kleene}
        & \textbf{Online}
        & \textbf{Sharing}\\
        & \textbf{closure}
        & \textbf{aggregation}
        & \textbf{decisions}
       \\\hline\hline
       \mcep~\cite{KS19}
       & \checkmark
       & $-$
       & static
       \\\hline
       \sharon~\cite{PRLRM18}
       & $-$
       & \checkmark 
       & static
       \\\hline
       \greta~\cite{PLRM18}
       & \checkmark
       & \checkmark 
       & not shared
       \\\hline
       \app\ (ours)
       & \checkmark
       & \checkmark 
       & dynamic
       \\\hline
    \end{tabular}
    \caption{Approaches to event trend aggregation}
    \label{tab:approaches}
\end{table}

\vspace*{-5mm}

\textbf{State-of-the-Art Approaches}. 
While there are approaches to shared execution of multiple Kleene queries~\cite{hong2009rule,KS19}, they first construct all trends and then aggregate them. Even if trend construction is shared, its exponential complexity is not avoided~\cite{PLAR17,ZDI14}. Thus, even the most recent approach, \mcep~\cite{KS19} is 76--fold slower than \app\ as the number of events scales to 10K events per window (Figure~\ref{fig:e1-latency-events}).
%
Recent work on event trend processing~\cite{PLRM18,PLRM19,PRLRM18} addresses this performance bottleneck by pushing the aggregation computation into the pattern matching process. Such online methods manage to skip the trend construction step and thus reduce time complexity of trend aggregation from exponential to quadratic in the number of matched events. Among these online approaches, \greta~\cite{PLRM18} is the only approach that supports Kleene closure. Unfortunately, \greta\ neglects sharing opportunities in the workload and instead processes each query independently from others.
On the other hand, while \sharon~\cite{PRLRM18} considers  sharing among  queries, it does not support Kleene closure. Thus, it is restricted to fixed-length event sequences. Further, its shared execution strategy is static and thus misses runtime sharing opportunities. Our experiments confirm that these existing approaches fail to cope with high velocity streams with 100K events per window (Figures~\ref{fig:e2-latency-nyc-events} and \ref{fig:e2-latency-sh-events}). Table~\ref{tab:approaches} summarizes the approaches mentioned above with respect to the challenges of shared execution of multiple trend aggregation queries.


\textbf{Proposed Solution}. 
To address these challenges, 
we now propose the \app\  approach that supports  online aggregation over Kleene closure while dynamically deciding which subset of sub-patterns should be shared by which trend aggregation queries and for how long depending on the current characteristics of the event stream.
The \app\ optimizer leverages these stream characteristics to estimate the runtime sharing benefit.  Based on the estimated benefit, it instructs the \app\ executor to switch between shared and non-shared execution strategies. Such fine-grained decisions allow \app\ to maximize the sharing benefit at runtime. 
%
The \app\ runtime executor propagates shared trend aggregates from previously matched events to newly matched events 
in an online fashion, i.e., without constructing event trends.

\textbf{Contributions}.
\app\ 
offers the following key innovations.

1. We present a novel framework \app\ for optimizing a workload of queries computing aggregation over Kleene pattern matches, called event trends. To the best of our knowledge, \app\ is the first to seamlessly integrate the power of online event trend aggregation and adaptive execution sharing among queries.


2. We introduce the \app\ graph to compactly capture  trends matched by queries in the workload. We partition the graph into smaller graphlets by event types and time. \app\ then selectively shares trend aggregation in some graphlets among multiple queries. 

3. We design a lightweight sharing benefit model to quantify the trade-off between the benefit of sharing and the overhead of maintaining the intermediate trend aggregates per query at runtime. 

4. Based on the benefit of sharing sub-patterns, we propose an adaptive sharing optimizer.
It selects a subset of queries among which it is beneficial to share this sub-pattern and determines the time interval during which this sharing remains beneficial.

5. Our experiments on several real world  stream  data sets  demonstrate that \app\ achieves up to five orders of magnitude performance improvement over state-of-the-art approaches.

\textbf{Outline.} 
Section~\ref{sec:basic} describes preliminaries.
Sections~\ref{sec:executor} and \ref{sec:runtime} describe the core \app\ techniques: online trend aggregation and dynamic sharing optimizer. 
We present experiments,
review related work  and 
conclude the paper in Sections ~\ref{sec:experiments}, 
~\ref{sec:related}, and 
~\ref{sec:conclusions}, respectively.

\section{Preliminaries}
\label{sec:basic}

\subsection{Basic Notions}
\label{sec:basic_notions}

Time is represented by a linearly ordered set of time points $(\mathbb{T},\leq)$, where $\mathbb{T} \subseteq \mathbb{Q^+}$ are the non-negative rational numbers. 
An \textbf{\textit{event}} $e$ is a data tuple describing an incident of interest to the application. An event $e$ has a time stamp $e.time \in \mathbb{T}$ assigned by the event source. 
An event $e$ belongs to a particular event type $E$, denoted \textit{e.type=E} and described by a schema that specifies the set of event attributes and the domains of their values. A specific attribute $\mathit{attr}$ of $E$ is referred to as $E.\mathit{attr}$.
Table~\ref{tab:notation} summarizes the notation.

Events are sent by event producers (e.g., vehicles and mobile devices) to an  \textbf{\textit{event stream}} $I$. We assume that events arrive in order by their time stamps. Existing approaches to handle out-of-order events can be applied~\cite{CGM10, LTSPJM08, LLGRC09, SW04}.

An event consumer (e.g., Uber stream analytics) continuously monitors the stream with \textbf{\textit{event queries}}. We adopt the commonly used query language and semantics from SASE~\cite{ADGI08, WDR06, ZDI14}.
The query workload in Figure~\ref{fig:queries} is expressed in this language. We assume that the workload is static. Adding or removing a query from a workload requires migration of the execution plan to a new workload which can be handled by existing approaches~\cite{KWF06, ZhuRH04}.

\begin{table}[!tb]
    \centering
    \begin{tabular}{|p{1.3cm}|p{11.3cm}|}
    \hline
        Notation 
        & Description \\\hline\hline
       $e.\mathit{time}$
       & Time stamp of event $e$ \\\hline
       $e.\mathit{type}$
       & Type of event $e$ \\\hline
       $E.\mathit{attr}$
       & Attribute $\mathit{attr}$ of event type $E$ \\\hline
       $\mathit{start}(q)$
       & Start types of the pattern of query $q$ \\\hline
       $\mathit{end}(q)$
       & End types of the pattern of query $q$ \\\hline
       $\mathit{pt}(E,q)$
       & Predecessor types of event type $E$ w.r.t query $q$ \\\hline
       $\mathit{pe}(e,q)$
       & Predecessor events of event $e$ w.r.t query $q$  \\\hline
       $n$ & Number of events per window \\\hline
       $g$ & Number of events per graphlet \\\hline
       $b$ & Number of events per burst \\\hline
       $k$ & Number of queries in the workload $Q$ \\\hline
       $k_s$ & Number of queries that share the graphlet $G_E$ with other queries \\\hline
       $k_n$ & Number of queries that do not share the graphlet $G_E$ with other queries \\\hline
       $p$
       & Number of predecessor types per type per query \\\hline
       $s$ & Number of snapshots \\\hline
       $s_c$
       & Number of snapshots created from one burst of events \\\hline
       $s_p$
       & Number of snapshots propagated in one shared graphlet \\\hline
    \end{tabular}
    \caption{Table of notations}
    \label{tab:notation}
\end{table}


\begin{definition}(\textbf{Kleene Pattern})
A pattern $P$ can be in the form of 
$E$, $P_1+$, ($\mynot\ P_1$),
\seq$(P_1,$ $P_2)$, $(P_1 \vee P_2)$, or $(P_1 \wedge P_2)$, 
where 
$E$ is an event type,
$P_1,P_2$ are patterns, 
$+$ is a Kleene plus, 
$\mynot$ is a negation,
\seq\ is an event sequence,
$\vee$ is a disjunction, and
$\wedge$ is a conjunction.
$P_1$ and $P_2$ are called sub-patterns of $P$.
If a pattern $P$ contains a Kleene plus operator, $P$ is called a Kleene pattern.
%
\label{def:pattern}
\end{definition}

%
%
%



\begin{definition}(\textbf{Event Trend Aggregation Query})
An event trend aggregation query $q$ consists of five clauses:

$\bullet$ Aggregation result specification (\textsf{RETURN} clause),

$\bullet$ Kleene pattern $P$ (\textsf{PATTERN} clause) as per Definition~\ref{def:pattern},

$\bullet$ Predicates $\theta$ (optional \textsf{WHERE} clause),

$\bullet$ Grouping $G$ (optional \textsf{GROUPBY} clause), and

$\bullet$ Window $w$ (\textsf{WITHIN/SLIDE} clause).
\label{def:query}
\end{definition}

\begin{definition}(\textbf{Event Trend})
Let $q$ be a query per Definition~\ref{def:query}.
An event trend $tr = (e_1, \ldots, e_k)$ corresponds to a sequence of events that conform to the pattern $P$ of $q$. All events in a trend $tr$ satisfy predicates $\theta$, have the same values of grouping attributes $G$, and are within one window $w$ of $q$. 
\label{def:trend}
\end{definition}

\textbf{Aggregation of Event Trends}.
Within each window specified by the query $q$, event trends are grouped by the values of grouping attributes $G$. Aggregates are then computed per group. 
\app\ focuses on distributive (\mycount, \mymin, \mymax, \mysum) and algebraic aggregation functions (\myavg) since they can be computed incrementally~\cite{Gray97}. 
Let $E$ be an event type,
$\mathit{attr}$ be an attribute of $E$, and
$e$ be an event of type $E$.
While $\mycount(*)$ returns the number of all trends per group,
$\mycount(E)$ computes the number of all events $e$ in all trends per group. 
$\mysum(E.\mathit{attr})$ ($\myavg(E.\mathit{attr})$) calculates the summation (average) of the value of $\mathit{attr}$ of all events $e$ in all trends per group.
$\mymin(E.\mathit{attr})$ ($\mymax(E.\mathit{attr})$) computes the minimal (maximal) value of $\mathit{attr}$ for all events $e$ in all trends per group.

\subsection{\app\ Approach in a Nutshell}
\label{sec:overview}

Given a workload of event trend aggregation queries $Q$ and a high-rate event stream $I$, the \textbf{\textit{Multi-query Event Trend Aggregation Problem}} is to evaluate the workload $Q$ over the stream $I$ such that the average query latency of all queries in $Q$ is minimal. The latency of a query $q \in Q$ is measured as the difference between the time point of the aggregation result output by the query $q$ and the arrival time of the last event that contributed to this result.

\begin{figure}[!ht]
\centering
\includegraphics[width=.6\columnwidth]{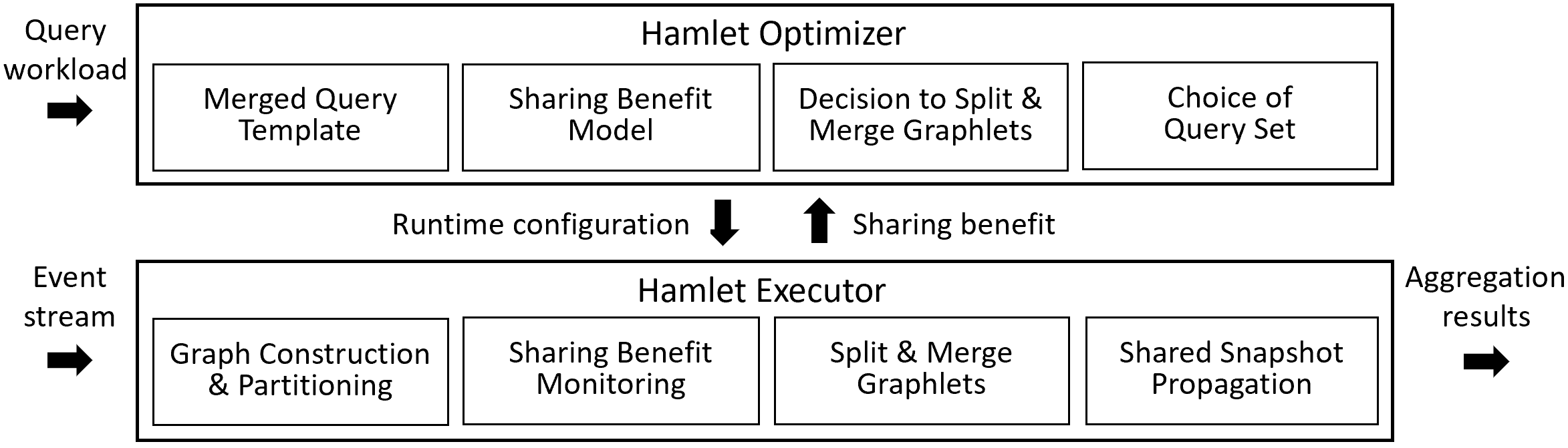}
\vspace{-5pt}
\caption{\app\ Framework}
\label{fig:framework}
\end{figure}

We design the \app\ framework in Figure~\ref{fig:framework} to tackle this problem.
To reveal all sharing opportunities in the workload at compile time, the \textbf{\textit{\app\ Optimizer}} identifies sharable queries and translates them into a Finite State Automaton-based representation, called the merged query template (Section~\ref{sec:tempalte}).
Based on this template, the optimizer reveals which sub-patterns could potentially be shared by which queries.
At runtime, the optimizer estimates the sharing benefit depending on the current stream characteristics to make fine-grained sharing decisions. Each sharing decision determines which queries share the processing of which Kleene sub-patterns and for how long (Section~\ref{sec:runtime}). 
These decisions along with the template are encoded into the runtime configuration to guide the executor.


\textbf{\textit{\app\ Executor}}  partitions the stream by the values of grouping attributes. To enable shared execution despite different windows of sharable queries, the executor further partitions the stream into panes that are sharable across overlapping windows~\cite{AW04, GSCL12, KWF06, LMTPT05}. Based on the merged query template for each set of sharable queries, the executor compactly encodes matched trends  
within a pane into the \app\ graph.
More precisely, matched events are modeled as  nodes, while event adjacency relations in a trend are edges of the graph.
%
%
Based on this graph, we incrementally compute trend aggregates by propagating intermediate aggregates along the edges from previously matched events to new events
-- without constructing the actual trends. 
This reduces the time complexity of trend aggregation from exponential to quadratic in the number of matched events compared to two-step approaches~\cite{hong2009rule,KS19,PLAR17,ZDI14}.

The \app\ graph is partitioned into sub-graphs, called graphlets, by event type and time stamps to maximally expose runtime opportunities to share these graphlets among  queries. 
Since the value of aggregates may be differ  for distinct
 queries, we capture these aggregate values per query as "snapshots" and share the propagation of snapshots through  shared graphlets (Section~\ref{sec:shared-approach}). 

Lastly, the executor implements the sharing decisions imposed by the optimizer. This may involve dynamically splitting a shared graphlet into several non-shared graphlets or, vice-versa, merging several non-shared graphlets into one shared graphlet (Section~\ref{sec:decisions}).

\section{Core \app\ Execution Techniques}
\label{sec:executor}

\textbf{Assumptions}.
To keep the discussion focused on the core concepts, we make simplifying assumptions in Sections~\ref{sec:executor} and \ref{sec:runtime}. We drop them to extend \app\ to the broad class of trend aggregation queries (Definition~\ref{def:query}) in Section~\ref{sec:other}.
These assumptions include:
(1)~queries compute the number of trends per window $\mycount(*)$;
(2)~query patterns do not contain disjunction, conjunction nor negation; and 
(3)~Kleene plus operator is applied to an event type  and appears once per query.


In Section~\ref{sec:tempalte}, we describe the workload and stream partitioning. We introduce strategies for processing  queries without sharing in Section~\ref{sec:non-shared-baseline} versus with \textit{shared} online trend aggregation in Section~\ref{sec:shared-approach}. In Section~\ref{sec:runtime}, we present the runtime optimizer that makes these sharing decisions.

\subsection{Workload Analysis and Stream Partitioning}
\label{sec:tempalte}

Given that the workload may contain queries with different Kleene patterns, aggregation functions, windows, and groupby clauses, \app\ takes the following pre-processing steps:
(1)~it breaks the workload  into  sets of sharable queries at compile time;
(2)~it then constructs the \app\ query template for each sharable query set; and
(3)~it partitions the stream by window and groupby clauses for each query template at runtime. 

\begin{definition}(\textbf{Shareable Kleene Sub-pattern})
Let $Q$ be a workload and $E$ be an event type. 
Assume that a Kleene sub-pattern $E+$ appears in queries $Q_E \subseteq Q$ and $|Q_E| > 1$. 
We say that $E+$ is shareable by queries $Q_E$.
\label{def:shareable-sub-pattern}
\end{definition}

However, sharable Kleene sub-patterns cannot always be shared due to other query clauses. For example, queries having $\mycount(*)$, $\mymin(E.$ $\mathit{attr})$ or $\mymax(E.\mathit{attr})$ can only be shared with queries that compute these same aggregates. 
In contrast,   
since $\myavg(E.\mathit{attr})$ is computed as $\mysum(E.\mathit{attr})$ divided by $\mycount(E)$, queries computing $\myavg(E.\mathit{attr})$ can be shared with queries that calculate $\mysum(E.\mathit{attr})$ or $\mycount(E)$.
We therefore define sharable queries below.

\begin{definition}(\textbf{Sharable Queries})
Two queries are \textit{sharable} if their patterns contain at least one sharable Kleene sub-pattern, their aggregation functions can be shared, their windows overlap, and their grouping attributes are the same.
\label{def:sharable_queries}
\end{definition}

To facilitate the shared runtime execution of each set of sharable queries, each pattern is converted into its Finite State Automaton-based representation~\cite{ADGI08, DGPRSW07, WDR06, ZDI14}, called \textbf{\textit{query template}}.
We adopt the state-of-the-art algorithm~\cite{PLRM18} to convert each pattern in in the workload $Q$ into its template.

Figure~\ref{fig:template} depicts the template of query $q_1$ with pattern $\seq(A,$ $B+)$.
States, shown as rectangles, represent event types in the pattern.
If a transition connects a type $E_1$ with a type $E_2$ in a template of a query $q$, then events of type $E_1$ precede events of type $E_2$ in a trend matched by $q$. $E_1$ is called a \textit{\textbf{predecessor type}} of $E_2$, denoted $E_1 \in \textit{pt}(E_2,q)$.
A state without ingoing edges is a \textit{\textbf{start type}}, and a state shown as a double rectangle is an \textit{\textbf{end type}} in a pattern. 

\begin{example}
In Figure~\ref{fig:template}, events of type $B$ can be preceded by events of types $A$ and $B$ in a trend matched by $q_1$, i.e., $\mathit{pt}(B,q_1)=\{A,B\}$. Events of type $A$ are not preceded by any events, $\mathit{pt}(A,q_1)$ $=\emptyset$. Events of type $A$ start trends and events of type $B$ end trends matched by $q_1$, i.e., $\mathit{start}(q_1)=\{A\}$ and $\mathit{end}(q_1)=\{B\}$.
\label{ex:template_one_query}
\end{example}

Our \app\ system processes the entire workload $Q$ instead of each query in isolation. To expose all sharing opportunities in $Q$, we convert the entire workload $Q$ into one \textit{\textbf{\app\ query template}}. It is constructed analogously to a query template with two additional rules. First, each event type is represented in the merged template only once. Second, each transition is labeled by the set of queries for which this transition holds.

\begin{example}
Figure~\ref{fig:merged_template} depicts the template for the workload $Q=\{q_1,q_2\}$ where query $q_1$ has pattern $\seq(A,B+)$ and query $q_2$ has pattern $\seq(C,B+)$.
The transition from $B$ to itself is labeled by two queries $q_1$ and $q_2$. This transition corresponds to the shareable Kleene sub-pattern $B+$ in these queries (highlighted in gray).
\label{ex:running_example}
\end{example}

%
The event stream is first partitioned by the grouping attributes.
To enable shared execution despite different windows of sharable queries, \app\ further partitions the stream into \textbf{\textit{panes}} that are sharable across  overlapping windows~\cite{AW04, GSCL12, KWF06, LMTPT05}.
The size of a pane is the greatest common divisor (gcd) of all window sizes and window slides. For example, for two windows $(\within\ 10\ min\ \slide$ $5\ min)$ and $(\within\ 15\ min$ $\slide\ 5$ $min)$, the gcd is 5 minutes. In this example, a pane contains all events per 5 minutes interval. For each set of sharable queries, we  apply the \app\ optimizer and executor within each pane.

\begin{figure*}[t]
\centering
\begin{minipage}{.15\textwidth}
\subfigure[Query $q_1$]{
  \includegraphics[width=1.\columnwidth]{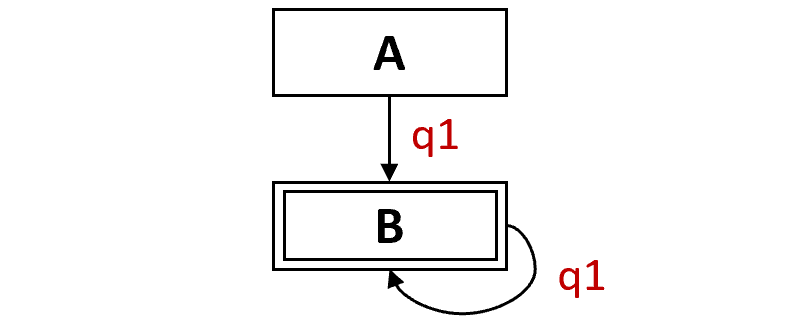}
\label{fig:template}
}
\subfigure[Workload $Q=\{q_1,q_2\}$]{
  \includegraphics[width=1.0\columnwidth]{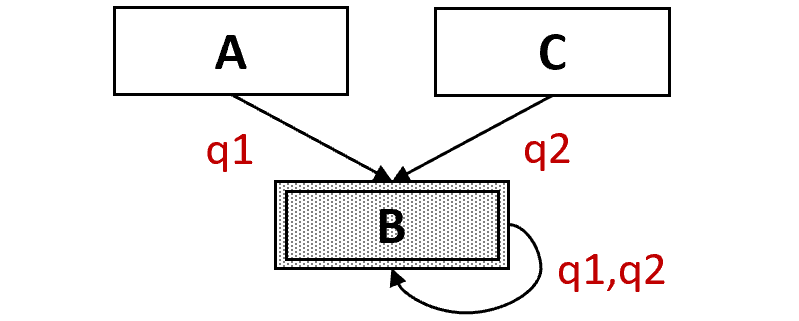}
\label{fig:merged_template}
}
\vspace{-5pt}
\caption{Template}
\label{fig:templates}
\end{minipage}
\hspace{5mm}
\begin{minipage}{.8\textwidth}
\subfigure[Non-shared \greta\ graph]{
\includegraphics[width=0.5\columnwidth]{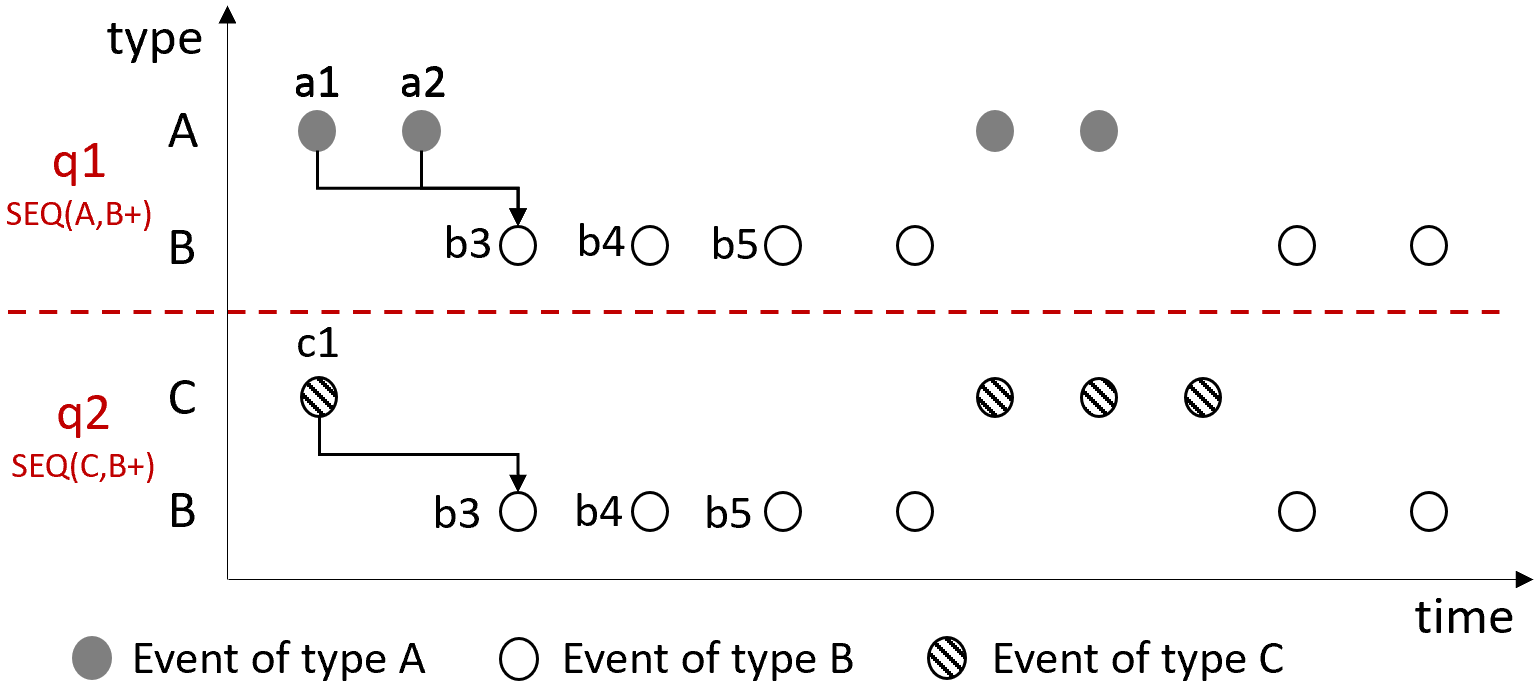}
\label{fig:not-shared}
}
\hspace{3mm}
\subfigure[Shared \app\ graph]{
\includegraphics[width=0.4\columnwidth]{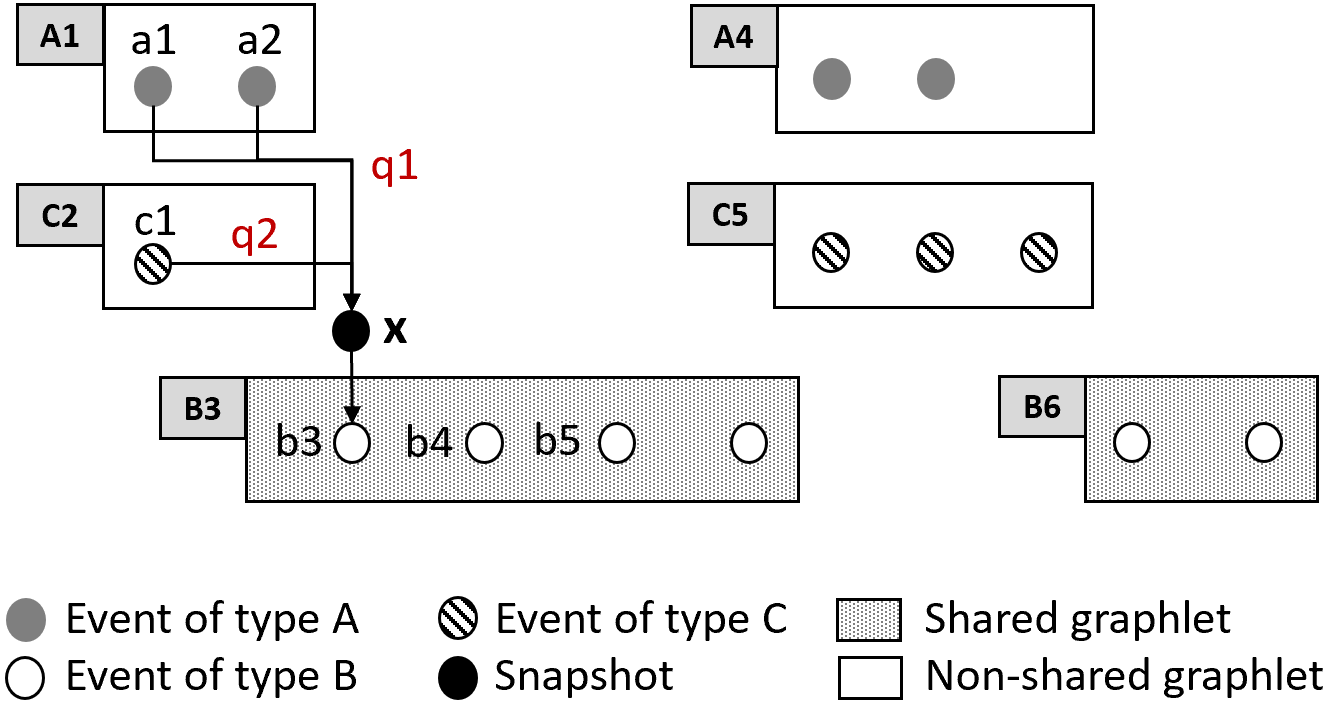}
\label{fig:snapshot}
}
\vspace{-5pt}
\caption{Non-shared vs shared execution}
\label{fig:greta_vs_hamlet}
\end{minipage}
\end{figure*}

\subsection{Non-Shared Online Trend Aggregation}
\label{sec:non-shared-baseline}

For the  non-shared execution, we describe below how the \app\ executor leverages state-of-the-art online trend aggregation approach
~\cite{PLRM18} to compute trend aggregates for each query {\it independently from all other queries}. Given a query $q$, it encodes all trends matched by $q$ in a query graph.
The nodes in the graph are events matched by $q$. Two events $e'$ and $e$ are connected by an edge if $e'$ and $e$ are adjacent in a trend matched by $q$. The event $e'$ is called a \textbf{\textit{predecessor event}} of $e$. At runtime, trend aggregates are propagated along the edges. In this way, 
we aggregate trends online, i.e., without actually constructing them.

Assume a query $q$ computes the number of trends $\mycount(*)$. 
When an event $e$ is matched by $q$, $e$ is inserted in the graph for $q$ and the \textbf{\textit{intermediate trend count}} of $e$ (denoted $\mathit{count}(e,q)$) is computed. $\mathit{count}(e,q)$ corresponds to the number of trends that are matched by $q$ and end at $e$. 
If $e$ is of start type of $q$, $e$ starts a new trend. Thus, $\mathit{count}(e,q)$ is incremented by one (Equation~\ref{eq:start_event}).
In addition, $e$ extends all trends that were previously matched by $q$. Thus, $\mathit{count}(e,q)$ is incremented by the sum of the intermediate trend counts of the predecessor events of $e$ that were matched by $q$ (denoted $\textit{pe}(e,q)$) (Equation~\ref{eq:event_count}). 
The \textbf{\textit{final trend count}} of $q$ is the sum of intermediate trend counts of all matched events of end type of $q$ (Equation~\ref{eq:final_count}).
%
\begin{align}
\mathit{start}(e,q) &=
    \begin{cases}
      1, & \text{if}\ \mathit{e.type} \in \mathit{start}(q) \\
      0, & \text{otherwise}
    \end{cases}
\label{eq:start_event}\\
\mathit{count}(e,q) &= 
\mathit{start}(e,q) + 
\sum_{e' \in \textit{pe}(e,q)}  \mathit{count}(e',q) 
\label{eq:event_count}\\
\mathit{fcount}(q) &= 
\sum_{\mathit{e.type} \in \textit{end}(q)}  \mathit{count}(e,q)
\label{eq:final_count}
\end{align}
\vspace{-5pt}

\begin{figure*}[t]
\centering
\subfigure[Snapshot $x$ at graphlet level]{
  \includegraphics[width=0.3\columnwidth]{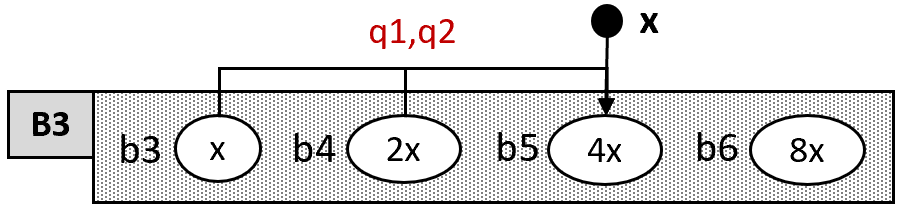}
\label{fig:no-predicates}
}
\hspace{2mm}
\subfigure[Snapshots $x$ and $y$ at graphlet level]{
  \includegraphics[width=0.3\columnwidth]{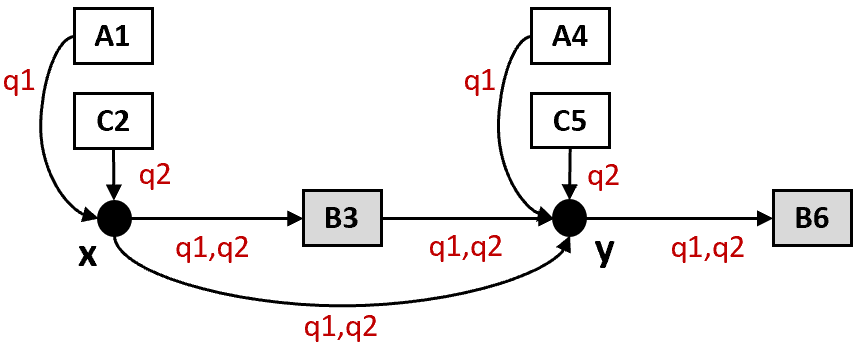}
\label{fig:snapshots}
}
\hspace{2mm}
\subfigure[Snapshot \textbf{\textit{z}} at event level]{
  \includegraphics[width=0.3\columnwidth]{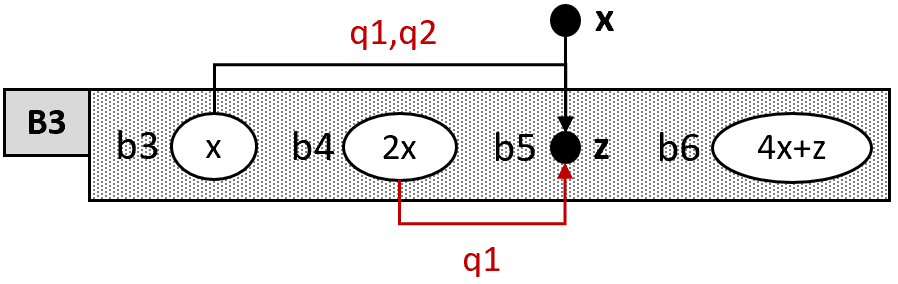}
\label{fig:predicates}
}
\vspace{-10pt}
\caption{Snapshots at graphlet and event levels}
\label{fig:snapshots_at_different_levels}
\end{figure*}

\begin{table*}[t]
    \begin{minipage}{.27\linewidth}
    \centering
    {\footnotesize\begin{tabular}{|c|p{2.5cm}|}
        \hline
        & \textbf{Trend count} 
        \\\hline\hline
        $b_3$ 
        & $x$ 
        \\\hline
        $b_4$ 
        & $x + count(b_3,Q) = 2x$ 
        \\\hline
        $b_5$ 
        & $x + count(b_3,Q) + count(b_4,Q) = 4x$ 
        \\\hline
        \multirow{2}{*}{$b_6$}
        & $x + count(b_3,Q) + count(b_4,Q) +$ 
        $count(b_5,Q) = 8x$ 
        \\\hline
    \end{tabular}}
    \caption{Shared propagation of \textbf{\textit{x}} within $\textbf{\textit{B}}_3$}
    \label{tab:snapshot}
    \end{minipage}
    \begin{minipage}{.35\linewidth}
    \centering
    {\footnotesize\begin{tabular}{|c|p{1.8cm}|p{1.8cm}|}
        \hline
        & \textbf{Query} $q_1$
        & \textbf{Query} $q_2$ 
        \\\hline\hline
        $x$ 
        & $sum(A_1,q_1) = 2$
        & $sum(C_2,q_2) = 1$ 
        \\\hline
        \multirow{4}{*}{$y$}
        & $value(x,q_1) +$ 
        $sum(B_3,q_1) +$  
        $sum(A_4,q_1) =$ 
        $2 + 15*2 + 2 = 34$
        & $value(x,q_2) +$ 
        $sum(B_3,q_2) +$
        $sum(C_5,q_2) =$ 
        $1 + 15*1 + 3 = 19$ 
        \\\hline
    \end{tabular}}
    \caption{Values of snapshots \textbf{\textit{x}} and \textbf{\textit{y}} per query}
    \label{tab:snapshots}
    \end{minipage} 
    \begin{minipage}{.35\linewidth}
    \centering
    {\footnotesize\begin{tabular}{|c|p{1.8cm}|p{1.8cm}|}
        \hline
        & \textbf{Query} $q_1$
        & \textbf{Query} $q_2$ 
        \\\hline\hline
        \multirow{3}{*}{$z$} 
        & $value(x,q_1) +$ 
        $count(b_3,q_1) +$ 
        $count(b_4,q_1) = 8$
        & $value(x,q_2) +$ 
        $count(b_3,q_2) = 2$ 
        \\\hline
        \multirow{3}{*}{$y$}
        & $value(x,q_1) +$ 
        $sum(B_3,q_1) +$ 
        $sum(A_4,q_1) = 34$
        & $value(x,q_2) +$ 
        $sum(B_3,q_2) +$ 
        $sum(C_5,q_2) = 15$ 
        \\\hline
    \end{tabular}}
    \caption{Values of snapshots \textbf{\textit{z}} and \textbf{\textit{y}} per query}
    \label{tab:snapshots2}
    \end{minipage} 
\end{table*}

\begin{example}
Continuing Example~\ref{ex:running_example}, a graph is maintained per each query in the workload $Q=\{q_1,q_2\}$ in Figure~\ref{fig:not-shared}. 
For readability, we sort all events by their types and timestamps. Events of types $A$, $B$, and $C$ are displayed as gray, white, and striped circles, respectively. We highlight the predecessor events of event $b_3$ by edges. All other edges are omitted for compactness. 
When $b_3$ arrives, two trends $(a_1,b_3)$ and $(a_2,b_3)$ are matched by $q_1$. Thus, $count$($b_3$,$q_1$) = $count(a_1,q_1) + count(a_2,q_1) = 2$. However, only one trend $(c_1,b_3)$ is matched by $q_2$. Thus, $count(b_3,q_2) = count(c_1,q_2) = 1$.
\end{example}


\textbf{Complexity Analysis}. 
Figure~\ref{fig:not-shared} illustrates that each event of type $B$ is stored and processed once for each query in the workload $Q$, introducing significant re-computation and replication overhead. 
Let $k$ denote the number of queries in the workload $Q$ and $n$  the number of events. Each query $q$ stores each matched event $e$ and computes the intermediate count of $e$ per Equation~\ref{eq:event_count}. All predecessor events of $e$ must be accessed, with $e$ having at most $n$ predecessor events.  
Thus, the time complexity of non-shared online trend aggregation is computed as follows:
\begin{equation}
\mathit{NonShared}(Q) = k \times n^2
\label{eq:nonshared-cost}
\end{equation}
\vspace{-3mm}

Events that are matched by $k$ queries are replicated $k$ times (Figure~\ref{fig:not-shared}). Each event stores its intermediate trend count. In addition, one final result is stored per query. Thus, the space complexity is $O(k \times n + k) = O(k \times n)$.





\subsection{Shared Online Trend Aggregation}
\label{sec:shared-approach}

In Equation~\ref{eq:nonshared-cost}, the overhead of processing each event once per query in the workload $Q$ is represented by the multiplicative factor $k$. 
Since the number of queries in a production workload may reach hundreds to thousands~\cite{ADLS, shared_clouds}, this re-computation overhead can be significant.
Thus, we design an efficient shared online trend aggregation strategy that encapsulates bursts of events of the same type in a graphlet such that the propagation of trend aggregates  within these graphlets can be shared among several queries. 


\begin{definition}(\textbf{Graphlet})
Let $q \in Q$ be a query and $T$ be a set of event types that appear in the pattern of $q$. 
A \textit{graphlet} $G_E$ is a graph of events of type $E$, if no events of type $E' \in T,\ E' \neq E,$ are matched by $q$ during the time interval $(e_\mathit{f}.time, e_l.time)$, where $e_\mathit{f}.time$ and $e_l.time$ are the time\-stamps of the first and the last events in $G_E$, respectively.
%
If new events can be added to a graphlet $G_E$ without violating the constraints above, the graphlet $G_E$ is called \textit{active}. Otherwise, $G_E$ is called \textit{inactive}.
\label{def:graphlet}
\end{definition}

\begin{definition}(\textbf{Shared Graphlet, \app\ Graph})
Let $E+$ be a Kleene sub-pattern that is shareable by queries $Q_E \subseteq Q$ (Definition~\ref{def:shareable-sub-pattern}). 
We call a graphlet $G_E$ of events of type $E$ a shared graphlet. 
The set of all interconnected shared and non-shared graphlets for a workload $Q$ is called a \app\ graph.
\label{def:shared-graphlet}
\end{definition}

\begin{example}
In Figure~\ref{fig:snapshot}, matched events are partitioned into six graphlets $A_1$--$B_6$ by their types and timestamps. For example, graphlets $B_3$ and $B_6$ are of type $B$. They are shared by queries $q_1$ and $q_2$.
In contrast to the non-shared strategy in Figure~\ref{fig:not-shared}, each event is stored and processed once for the entire workload $Q$. 
Events in $A_1$--$C_2$ are predecessors of events in $B_3$, while events in $A_1$--$C_5$ are predecessors of events in $B_6$. 
For readability, only the predecessor events of $b_3$ are highlighted by edges in Figure~\ref{fig:snapshot}. All other edges are omitted. 
$a_1$ and $a_2$ are predecessors of $b_3$ only for $q_1$, while $c_1$ is a predecessor of $b_3$ only for $q_2$. 
\label{ex:three}
\end{example}

Example~\ref{ex:three} illustrates the following two challenges of online shared event trend aggregation.

\textit{\textbf{Challenge 1}}. 
Given that event $b_3$ has different predecessors for queries $q_1$ and $q_2$, the computation of the intermediate trend count of $b_3$ (and all other events in graphlets $B_3$ and $B_6$) cannot be directly shared by queries $q_1$ and $q_2$.

\textit{\textbf{Challenge 2}}. 
If queries $q_1$ or $q_2$ have predicates, then not all previously matched events are qualified to contribute to the trend count of a new event. Assume that the edge between events $b_4$ and $b_5$ holds for $q_1$ but not for $q_2$ due to predicates, and all other edges hold for both queries. Then $count(b_4,q_1)$ contributes to $count(b_5,q_1)$, but $count(b_4,q_2)$ does not contribute to $count(b_5,q_2)$.

We tackle these challenges by introducing \textbf{\textit{snapshots}}.
Intuitively, a snapshot is a variable that its value corresponds to an intermediate trend aggregate per query. 
In Figure~\ref{fig:snapshot}, the propagation of a snapshot $x$ within graphlet $B_3$ is shared by queries $q_1$ and $q_2$. 
We store the values of $x$ per query (e.g., $x=2$ for $q_1$ and $x=1$ for $q_2$).

\begin{definition}(\textbf{Snapshot at Graphlet Level})
Let $E'$ and $E$ be distinct event types.
Let $E+$ be a Kleene sub-pattern that is shared by queries $Q_E \subseteq Q$, $q \in Q_E$.
Let $E' \in \mathit{pt}(E,q)$ and $G_{E'}$ and $G_E$ be graphlets of events of types $E'$ and $E$, respectively.
Assume for any events $e' \in G_{E'}, e \in G_E$, $e'.time < e.time$ holds. 
A snapshot $x$ of the graphlet $G_{E'}$ is a variable whose value is computed per query $q$ and corresponds to the intermediate trend count of the query $q$ at the end of the graphlet $G_{E'}$. 
\begin{equation}
\mathit{value}(x,q) = \mathit{sum}(G_{E'},q) = \sum_{e' \in G_{E'}} \mathit{count}(e',q)
\label{eq:snapshot}
\end{equation}

The propagation of snapshot $x$ through the graphlet $G_E$ follows Equation~\ref{eq:event_count} and is shared by queries $Q_E$.
\label{def:snapshot}
\end{definition}

\begin{example}
When graphlet $B_3$ starts, a snapshot $x$ is created. $x$ captures the intermediate trend count of query $q_1$ ($q_2$) based on the intermediate trend counts of all events in graphlet $A_1$ ($C_2$). $x$ is propagated through graphlet $B_3$ as shown in Figure~\ref{fig:no-predicates} and Table~\ref{tab:snapshot}. 

Analogously, when graphlet $B_6$ starts, a new snapshot $y$ is created. The value of $y$ is computed for queries $q_1$ ($q_2$) based on the value of $x$ for $q_1$ ($q_2$) and graphlets $B_3$ and $A_4$ ($C_5$). Figure~\ref{fig:snapshots} illustrates the connections between snapshots and graphlets. The edges from graphlets $A_1$ and $A_4$ ($C_2$ and $C_5$) hold only for query $q_1$ ($q_2$). Other edges hold for both queries $q_1$ and $q_2$.

Table~\ref{tab:snapshots} captures the values of snapshots $x$ and $y$ per query. For compactness, $sum(A_1,q_1)$ denotes the sum of intermediate trend counts of all events in $A_1$ that are matched by $q_1$ (Equation~\ref{eq:snapshot}). 
When the snapshot $y$ is created, the value of $x$ per query is plugged in to obtain the value of $y$ per query. The propagation of $y$ through $B_6$ is shared by $q_1$ and $q_2$. In this way, only one snapshot is propagated at a time to keep the overhead of snapshot maintenance low.
\end{example}

To enable shared trend aggregation despite expressive predicates, we now introduce snapshots at the event level.

\begin{definition}(\textbf{Snapshot at Event Level})
Let $G_E$ be a graphlet that is shared by queries $Q_E \subseteq Q$. Let $q_1,q_2 \in Q_E$ and $e_1, e_2 \in G_E$ be events such that the edge $(e_1,e_2)$ holds for $q_1$ but does not hold for $q_2$ due to predicates. 
A snapshot $z$ is the intermediate trend count of $e_2$ that is 
computed for $q_1$ and $q_2$ per Equation~\ref{eq:event_count} and propagated through the graphlet $G_E$ for all queries in $Q_E$.
\label{def:snapshot2}
\end{definition}
%


\begin{example}
In Figure~\ref{fig:predicates}, assume that the edge between events $b_4$ and $b_5$ holds for query $q_1$ but not for query $q_2$ due to predicates. All other edges hold for both queries. Then, $count(b_4,q_1)$ contributes to $count(b_5,q_1)$, but $count(b_4,q_2)$ does not contribute to $count(b_5,q_2)$. To enable shared processing of graphlet $B_3$ despite predicates, we introduce a new snapshot $z$ as the intermediate trend count of $b_5$ and propagate both snapshots $x$ and $z$ within graphlet $B_3$. 
Table~\ref{tab:snapshots2} summarizes the values of $z$ and $y$ per query. 
\end{example}

\textbf{Shared Online Trend Aggregation Algorithm} computes the number of trends per query $q \in Q$ in the stream $I$. For simplicity, we assume that the stream $I$ contains events within one pane.
For each event $e \in I$ of type $E$, Algorithm~\ref{algo:snapshot-propagation} constructs the \app\ graph and computes the trend count as follows. 

\textbf{\textit{\app\ graph construction}} (Lines~4--14).
When an event $e$ of type $E$ is matched by a query $q \in Q$, $e$ is inserted into a graphlet $G_E$ that stores events of type $E$ (Line~14). 
if there is no active graphlet $G_E$ of events of type $E$, we create a new graphlet $G_E$, mark it as active and store it in the \app\ graph $G$ (Lines~7--8). If the graphlet $G_E$ is shared by queries $Q_E \subseteq Q$, then we create a snapshot $x$ at graphlet level (Line~9). $x$ captures the values of intermediate trend counts per query per Equation~\ref{eq:snapshot} at the end of graphlet $G_{E'}$ that stores events of type $E',\ E' \in pt(E,q)$. We save the value of $x$ per query in the table of snapshots $S$ (Lines~10--13).
Also, for each query $q \in Q$ with event types $T$, we mark all graphlets $G_{E'}$ of events of type $E' \in T,\ E' \neq E,$ as inactive (Lines~4--6). 

\textbf{\textit{Trend count computation}} (Lines~16--24).
If $G_E$ is shared by queries $Q_E \subseteq Q$ and the set of predecessor events of $e$ is identical for all queries $q \in Q_E$, then we compute $count(e,q)$ per Equation~\ref{eq:event_count} (Lines~16--18).
If $G_E$ is shared but the sets of predecessor events of $e$ differ among the different queries in $Q_E$ due to predicates, then we create a snapshot $y$ as the intermediate trend count of $e$ (Line~19). 
We compute the value of $y$ for each query $q \in Q_E$ per Equation~\ref{eq:event_count} and save it in the table of snapshots $S$ (Line~20).
If $G_E$ is not shared, the algorithm defaults to the non-shared trend count propagation per Equation~\ref{eq:event_count} (Line~21).
If $E$ is an end type for a query $q \in Q$, we increment the final trend count of $q$ in the table of results $R$ by the intermediate trend count of $e$ for $q$ per Equation~\ref{eq:final_count} (Lines~22--23).
Lastly, we return the table of results $R$ (Line~24).

\begin{theorem}
Algorithm~\ref{algo:snapshot-propagation} returns correct event trend count for each query in the workload $Q$. 
\end{theorem}

\begin{proof}[Proof Sketch]
Correctness of the graph construction for a single query and the non-shared trend count propagation through the graph as defined in Equation~\ref{eq:event_count} are proven in~\cite{PLRM18}. 
Correctness of the snapshot computation per query as defined in Equation~\ref{eq:snapshot} follows from Equation~\ref{eq:event_count}. 
Algorithm~\ref{algo:snapshot-propagation} propagates snapshots through the \app\ graph analogously to trend count propagation through the \greta\ graph defined in~\cite{PLRM18}.
\end{proof}

\textbf{Data Structures}.
Algorithm~\ref{algo:snapshot-propagation} utilizes the following physical data structures.

(1) \textbf{\textit{\app\ graph}} $G$ is a set of all graphlets. Each graphlet has two metadata flags \textit{active} and \textit{shared} (Definitions~\ref{def:graphlet} and \ref{def:shared-graphlet}). 

(2) \textbf{\textit{A hash table of snapshot coefficients}} per event $e$. The intermediate trend count of $e$ may be an expression composed of several snapshots. 
In Figure~\ref{fig:predicates}, $count(b_6,Q) = 4x + z$. 
Such composed expressions are stored in a hash table per event that maps a snapshot to its coefficient. In this example, $x \mapsto 4$ and $z \mapsto 1$ for $b_6$.

(3) \textbf{\textit{A hash table of snapshots}} $S$ is a mapping from a snapshot $x$ and a query $q$ to the value of $x$ for $q$ (Tables~\ref{tab:snapshots} and \ref{tab:snapshots2}). 

(4) \textbf{\textit{A hash table of trend count results}} $R$ is a mapping from a query $q$ to its corresponding trend count.

\begin{algorithm}[!ht]
\begin{algorithmic}[1]
\Require Query workload $Q$, event stream $I$, \app\ graph $G$, hash table of snapshots $S$
\Ensure Hash table of results $R$ 
\State $G \leftarrow \emptyset$, $S, R \leftarrow$ empty hash tables
\ForAll {event $e \in I$ with $e.type=E$} 
    \State $//$ \textbf{\app\ graph construction}
    \ForAll {$q \in Q$ \text{ with event types }T}
        \ForAll {$E' \in T,\ E' \neq E$}
            \State $G_{E'} \leftarrow \mathit{getGraphlet}(G,E')$,
            $G_{E'}.\mathit{active} \leftarrow \mathit{false}$
        \EndFor
    \EndFor
    \If {\textbf{not} $G_E.\mathit{active}$}
        \State $G_E \leftarrow \mathit{createGraphlet()}$, $G_{E}.\mathit{active} \leftarrow \mathit{true}$,
        $G \leftarrow G \cup G_E$
        \If {$G_E.\mathit{shared}$ by $Q_E \subseteq Q$}
            $x \leftarrow \mathit{createSnapshot()}$ 
            \ForAll {$q \in Q_E$}
                \ForAll{$E' \in \mathit{pt}(E,q), E' \neq E$}
                    \State $G_{E'} \leftarrow \mathit{getGraphlet}(G,E')$
                    \State $S(x,q) \leftarrow S(x,q) + sum(G_{E'},q)$ \hspace{0.5cm}$//$ Eq.~5
                \EndFor
            \EndFor
        \EndIf    
    \EndIf
    \State insert $e$ into $G_E$
    \State $//$ \textbf{Trend count computation}
    \If {$G_E.\mathit{shared}$ by $Q_E \subseteq Q$}
        \If {$\forall q \in Q_E\ pe(e,q)$ are identical}
            \State $count(e,Q_E) \leftarrow count(e,q)$ \hspace{2.3cm}$//$ Eq.~2
        \Else\ $y \leftarrow \mathit{createSnapshot()}$, $count(e,Q_E) = y$
            \ForAll {$q \in Q_E$}
                $S(y,q) \leftarrow count(e,q)$ \hspace{0.2cm}$//$ Eq.~2
            \EndFor
          \EndIf
    \Else\ $count(e,q)$ \hspace{5.2cm}$//$ Eq.~2
    \EndIf
    \ForAll{$q \in Q$}
  	    \If {$E \in \mathit{end}(q)$} 
  		    $R(q) \leftarrow R(q) + count(e,q)$ $//$ Eq.~3
        \EndIf
    \EndFor
\EndFor
\State \Return $R$
\end{algorithmic}
\caption{\app\ shared online trend aggregation}
\label{algo:snapshot-propagation}
\end{algorithm}

\textbf{Complexity Analysis}.
%
We use the notations in Table~\ref{tab:notation} and Algorithm~\ref{algo:snapshot-propagation}.  
For each event $e$ that is matched by a query $q \in Q$, Algorithm~\ref{algo:snapshot-propagation} computes the intermediate trend count of $e$ in an online fashion. This requires access to all predecessor events of $e$. In the worst case, $n$ previously matched events are the predecessor events of $e$. Since the intermediate trend count of $e$ can be an expression that is composed of $s$ snapshots, the intermediate trend count of $e$ is stored in the hash table that maps snapshots to their coefficients. Thus, the time complexity of intermediate trend count computation is $O(n \times s)$. In addition, the final trend count is updated per query $q$ if $E$ is an end type of $q$ in $O(k \times s)$ time. In summary, the time complexity of trend count computation is $O(n \times (n \times s + k \times s)) = O(n^2 \times s)$ since $n \geq k$.


In addition, Algorithm~\ref{algo:snapshot-propagation} maintains snapshots to enable shared trend count computation.
To compute the values of $s$ snapshots for each query $q$ in the workload of $k$ queries, the algorithm accesses $g$ events in $t$ graphlets $G_{E'}$ of events of type $E' \in T,\ E' \neq E$. Thus, the time complexity of snapshot maintenance is $O(s \times k \times g \times t)$. 
In summary,  time complexity of Algorithm~\ref{algo:snapshot-propagation} is computed as follows:
\begin{equation}
\mathit{Shared}(Q) = n^2 \times s + s \times k \times g \times t
\label{eq:shared-cost}
\end{equation}

Algorithm~\ref{algo:snapshot-propagation} stores each matched event in the \app\ graph once for the entire workload.
Each shared event stores a hash table of snapshot coefficients.
Each non-shared event stores its intermediate trend count.
In addition, the algorithm stores snapshot values per query. 
Lastly, the algorithm stores one final result per query.
Thus, the space complexity is $O(n + n \times s + s \times k + k) = O(n \times s + s \times k)$.

\section{Dynamic Sharing Optimizer}
\label{sec:runtime}


We first model the runtime benefit of sharing trend aggregation  (Section~\ref{sec:dynamic-benefit}). Based on this benefit model, our \app\ optimizer makes runtime sharing decisions for a given set of queries (Section~\ref{sec:decisions}).
Lastly, we describe how to choose a set of queries that share a Kleene sub-pattern (Section~\ref{sec:queries}).


\subsection{Dynamic Sharing Benefit Model}
\label{sec:dynamic-benefit}

On the up side, shared trend aggregation avoids the re-computation overhead for each query in the workload. 
On the down side, it introduces overhead to maintain snapshots. 
Next, we quantify the trade-off between shared versus non-shared execution.

Equations~\ref{eq:nonshared-cost} and \ref{eq:shared-cost} determine the cost of non-shared and shared strategies of all events within the  window for the entire workload $Q$ based on stream statistics.
In contrast to these coarse-grained static decisions, the \app\ optimizer makes \textbf{\textit{fine-grained runtime decisions}} for each burst of events for a sub-set of queries $Q_E \subseteq Q$. 
Intuitively, a burst is a set of consecutive events of type $E$, the processing of which can be shared by  queries $Q_E$ that contain a $E+$ Kleene sub-pattern. The \app\ optimizer decides at runtime if sharing a burst is beneficial. In this way, beneficial sharing opportunities are harvested for each burst at runtime.



\begin{definition}(\textbf{Burst of Events})
Let $E+$ be a sub-pattern that is sharable by queries $Q_E$.
Let $T$ be the set of event types that appear in the patterns of queries $Q_E$, $E \in T$. 
A set of events of type $E$ within a pane is called a \textit{burst} $B_E$, if no events of type $E' \in T,\ E' \neq E,$ are matched by the queries $Q_E$ during the time interval $(e_\mathit{f}.time, e_l.time)$, where $e_\mathit{f}.time$ and $e_l.time$ are the time\-stamps of the first and the last events in $B_E$, respectively.
If no events can be added to a burst $B_E$ without violating the above constraints, the burst $B_E$ is called \textit{complete}.
\label{def:burst}
\end{definition}

Within each pane, events that belong to the same burst are buffered until a burst is complete. The arrival of an event of type $E'$ or the end of the pane indicates that the burst is complete. In the following, we refer to complete bursts as bursts for compactness.

\app\ restricts event types in a burst for the following reason. Assuming that a burst contained an event $e$ of type $E'$, the event $e$ could be matched by one query $q_1$ but not by another query $q_2$ in $Q_E$. Snapshots would have to be introduced to differentiate between the aggregates of $q_1$ and $q_2$ (Section~\ref{sec:shared-approach}). Maintenance of these snapshots may reduce the benefit of sharing. Thus, the previous sharing decision may have to be reconsidered as soon as the first event arrives that is matched  by some queries in $Q_E$.


\begin{definition}(\textbf{Dynamic Sharing Benefit})
Let $E+$ be a Kleene sub-pattern that is shareable by queries $Q_E$,
$B_E$ be a burst of events of type $E$,
$b$ be the number of events in $B_E$,
$s_c$ be the number of snapshots that are created from this burst $B_E$, and
$s_p$ be the number of snapshots that are propagated to compute the intermediate trend counts for the burst $B_E$. 
Let $G_E$ denote a shared graphlet and $G_E^i$ denote a set of non-shared graphlets (one graphlet per each query in $Q_E$).
Other notations are consistent with previous sections (Table~\ref{tab:notation}).

The \textit{benefit} of sharing a graphlet $G_E$ by the queries $Q_E$ is computed as the difference between the cost of the non-shared and shared execution of the burst $B_E$.

\vspace{-4mm}
\begin{align}
&\mathit{Shared}(G_E,Q_E) 
= b \times n \times s_p
+ s_c \times k \times g \times t
\nonumber\\
&\mathit{NonShared}(G_E^i,Q_E) 
= k \times b \times n
\nonumber\\
&\mathit{Benefit}(G_E,Q_E) 
= \mathit{NonShared}(G_E^i,Q_E)
- \mathit{Shared}(G_E,Q_E)
\label{eq:dynamic-benefit}
\end{align}

If $\mathit{Benefit}(G_E,Q_E)>0$, then it is beneficial to share trend aggregation within the graphlet $G_E$ by the queries $Q_E$.
\label{def:dynamic-benefit}
\end{definition}

Based on Definition~\ref{def:dynamic-benefit}, we conclude that the more queries $k$ share trend aggregation, the more events $g$ are in shared graphlets, and the fewer snapshots $s_c$ and $s_p$ are maintained at a time, the higher the benefit of sharing will be. Based on this conclusion, our dynamic \app\ optimizer decides to share or not to share online trend aggregation (Section~\ref{sec:decisions}).




\begin{definition}(\textbf{Dynamic Sharing Benefit})
Let $E+$ be a Kleene sub-pattern that is shareable by queries $Q_E$,
$b$ be the number of events of type $E$ in a burst,
$s_c$ be the number of snapshots that are created from this burst, and
$s_p$ be the number of snapshots that are propagated to compute the intermediate trend counts for the burst. 
Let $G_E$ denote a shared graphlet and $G_E^i$ denote a set of non-shared graphlets (one graphlet per each query in $Q_E$).
Other notations are consistent with previous sections (Table~\ref{tab:notation}).

The \textit{benefit} of sharing a graphlet $G_E$ by the queries $Q_E$ is computed as the difference between the cost of the non-shared and shared execution of the event burst.

\vspace{-3mm}
\begin{align}
\mathit{Shared}(G_E,Q_E) 
&= s_c \times k \times g \times p 
+ b \times (\log_2(g) + n \times s_p) \nonumber\\
\mathit{NonShared}(G_E^i,Q_E) 
&= k \times b \times (\log_2(g) + n) \nonumber\\
\mathit{Benefit}(G_E,Q_E) 
&= \mathit{NonShared}(G_E^i,Q_E)
- \mathit{Shared}(G_E,Q_E)
\label{eq:dynamic-benefit}
\end{align}

If $\mathit{Benefit}(G_E,Q_E)>0$, then it is beneficial to share trend aggregation within the graphlet $G_E$ by the queries $Q_E$.
\label{def:dynamic-benefit}
\end{definition}

Based on Definition~\ref{def:dynamic-benefit}, we conclude that the more queries $k$ share trend aggregation, the more events $g$ are in shared graphlets, and the fewer snapshots $s_c$ and $s_p$ are maintained at a time, the higher the benefit of sharing will be. Based on this conclusion, our dynamic \app\ optimizer decides to share or not to share online trend aggregation (Section~\ref{sec:decisions}).




\begin{figure*}[t]
\centering
\subfigure[{\tiny Shared $B_3$}]{
  \includegraphics[width=0.13\columnwidth]{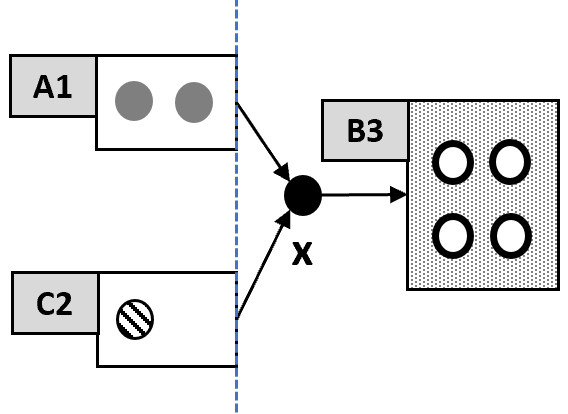}
\label{fig:benefit-monitoring-shared}
}
\subfigure[{\tiny Non-shared~$B_3$}]{
  \includegraphics[width=0.12\columnwidth]{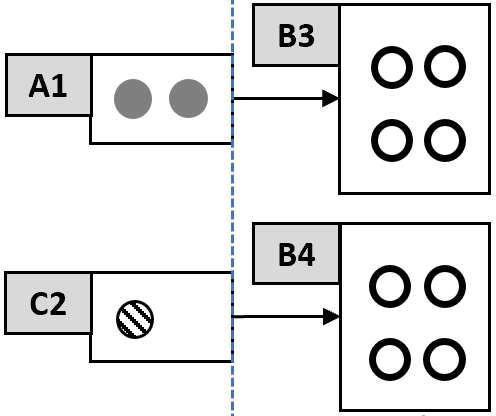}
\label{fig:benefit-monitoring-non-shared}
}
%
\subfigure[{\tiny Shared $B_3$}]{
  \includegraphics[width=0.12\columnwidth]{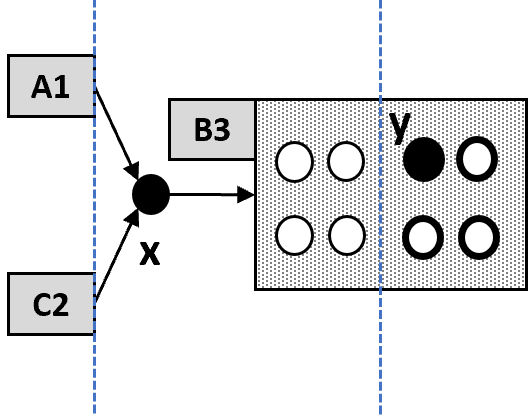}
\label{fig:decision-not-to-share-shared}
}
\subfigure[{\tiny Non-shared~$B_4,B_5$}]{
  \includegraphics[width=0.16\columnwidth]{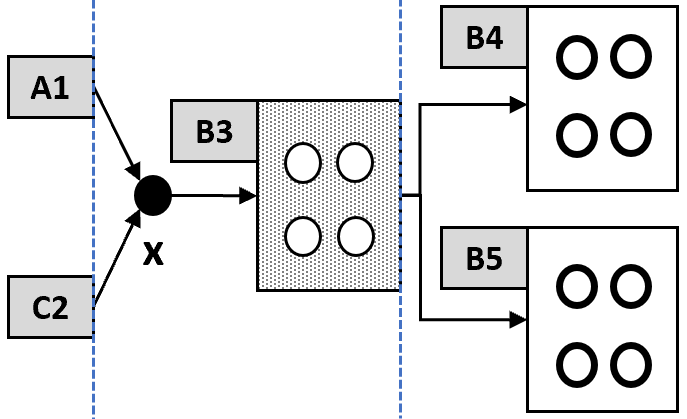}
\label{fig:decision-not-to-share-non-shared}
}
%
\subfigure[{\tiny Non-shared~$B_4,B_5$}]{
  \includegraphics[width=0.16\columnwidth]{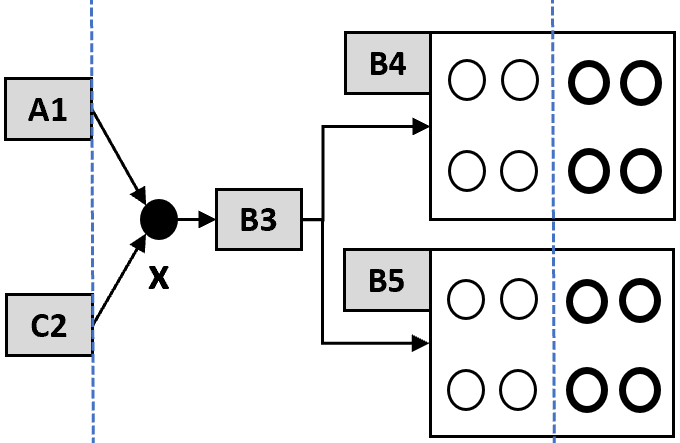}
\label{fig:decision-to-share-non-shared}
}
\subfigure[{\tiny Shared $B_6$}]{
  \includegraphics[width=0.21\columnwidth]{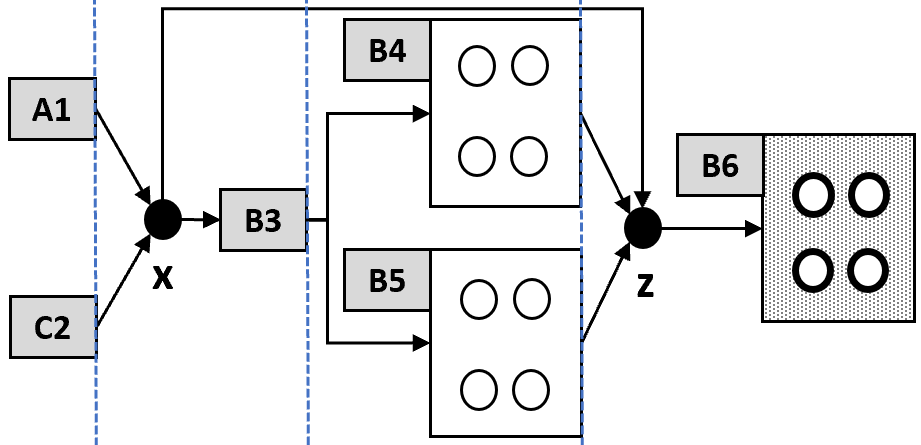}
\label{fig:decision-to-share-shared}
}
\caption{Dynamic sharing decisions. 
Decision to merge $B_3$ in (a) and (b).
Decision to split $B_3$ in (c) and (d).
Decision to merge $B_6$ in (e) and (f).}
\label{fig:desicions}
\end{figure*}

\subsection{Decision to Split and Merge Graphlets}
\label{sec:decisions}

Our dynamic \app\ optimizer monitors the sharing benefit depending on changing stream conditions at runtime.
Let $B+$ be a sub-pattern sharable by queries $Q_B = \{q_1,q_2\}$.
In Figure~\ref{fig:desicions}, 
pane boundaries are depicted as dashed vertical lines
and bursts of newly arrived events of type $B$ are shown as bold empty circles. For each burst, the optimizer has a choice of sharing (Figure~\ref{fig:benefit-monitoring-shared}) versus not sharing (Figure~\ref{fig:benefit-monitoring-non-shared}). It concludes that it is beneficial to share based on calculations in Equation~\ref{eq:benefit-monitoring}.
 
\vspace{-4mm}
\begin{align}
&\mathit{Shared}(B_3,Q_B) 
= 4 \times 7 \times 1 
+ 1 \times 2 \times 4 \times 2
= 44
\nonumber\\
&\mathit{NonShared}(\{B_3,B_4\},Q_B) 
= 2 \times 4 \times 7 = 56
\nonumber\\
&\mathit{Benefit}(B_3,Q_B)
= 56 - 44  = 12 > 0
\label{eq:benefit-monitoring}
\end{align}

\textbf{Decision to Split}.
However, when the next burst of events of type $B$ arrives, a new snapshot $y$ has to be created due to predicates during the shared execution in Figure~\ref{fig:decision-not-to-share-shared}. 
In contrast, the non-shared strategy processes queries $q_1$ and $q_2$ independently from each other (Figure~\ref{fig:decision-not-to-share-non-shared}). 
Now the overhead of snapshot maintenance is no longer justified by the benefit of sharing (Equation~\ref{eq:decision-not-to-share}).

\vspace{-4mm}
\begin{align}
&\mathit{Shared}(B_3,Q_B) 
= 4 \times 11 \times 2 
+ 1 \times 2 \times 8 \times 2
= 120
\nonumber\\
&\mathit{NonShared}(\{B_4,B_5\},Q_B) 
= 2 \times 4 \times 11 = 88
\nonumber\\
&\mathit{Benefit}(B_3,Q_B)
= 88 - 120 = -32 < 0
\label{eq:decision-not-to-share}
\end{align}

Thus, the optimizer decides to split the shared graph\-let $B_3$ into two non-shared graphlets $B_4$ and $B_5$ for the queries $q_1$ and $q_2$ respectively in Figure~\ref{fig:decision-not-to-share-non-shared}. Newly arriving events of type $B$ then must be inserted into both graphlets $B_4$ and $B_5$. Their intermediate trend counts are computed separately for the queries $q_1$ and $q_2$. The snapshot $x$ is replaced by its value for the query $q_1$ ($q_2$) within the graphlet $B_4$ ($B_5$). The graphlets $A_1$ and $C_2$ are collapsed.

\textbf{Decision to Merge}.
When the next burst of events of type $B$ arrives, we could either continue the non-shared trend count propagation within $B_4$ and $B_5$ (Figure~\ref{fig:decision-to-share-non-shared}) or merge $B_4$ and $B_5$ into a new shared graphlet $B_6$ (Figure~\ref{fig:decision-to-share-shared}). The \app\ optimizer concludes that the latter option is more beneficial in Equation~\ref{eq:decision-to-share}. As a consequence, a new snapshot $z$ is created as input to $B_6$. $z$ consolidates the intermediate trend counts of the snapshot $x$ and the graphlets $B_3$--$B_5$ per query $q_1$ and $q_2$.

\vspace{-4mm}
\begin{align}
&\mathit{Shared}(B_6,Q_B) 
= 4 \times 15 \times 1 
+ 1 \times 2 \times 4 \times 2 
= 76
\nonumber\\
&\mathit{NonShared}(\{B_4,B_5\},Q_B) 
= 2 \times 4 \times 15 = 120 
\nonumber\\
&\mathit{Benefit}(B_6,Q_B)
= 120 - 76 = 44 > 0
\label{eq:decision-to-share}
\end{align}


\textbf{Complexity Analysis}.
The runtime sharing decision per burst has constant time complexity because it simply plugs in locally available stream statistics into Equation~\ref{eq:dynamic-benefit}.
A graphlet split comes for free since we simply continue graph construction per query (Figure~\ref{fig:decision-not-to-share-non-shared}).
Merging graphlets requires creation of one snapshot and calculation of its values per query (Figure~\ref{fig:decision-to-share-shared}). Thus, the time complexity of merging is $O(k \times g \times t)$ (Equation~\ref{eq:shared-cost}).
Since our workload is fixed (Section~\ref{sec:basic}), the number of queries $k$ and the number of types $t$ per query are constants. Thus, the time complexity of merge is linear in the number of events per graphlet $g$.
Merging graphlets  requires storing the value of one snapshot per query. Thus, its space complexity is $O(k)$.


\subsection{Choice of Query Set}
\label{sec:queries}

To relax the assumption from Section~\ref{sec:decisions} that a set of queries $Q_E$ that share a Kleene sub-pattern $E+$ is given, we now select a sub-set of queries $Q_E$ from the workload $Q$ for which sharing $E+$ is the most beneficial among all other sub-sets of $Q$.
%
In general, the search space of all sub-sets of $Q$ is exponential in the number of queries in $Q$ since all combinations of shared and non-shared queries in $Q$ are considered. 
For example, if $Q$ contains four queries, Figure~\ref{fig:search-space} illustrates the search space of 12 possible execution plans of $Q$. 
Groups of queries in braces are shared. For example, the plan (134)(2) denotes that queries 1, 3, 4 share their execution, while query 2 is processed separately.
The search space ranges from maximally shared (top node) to non-shared (bottom node) plans. 
Each plan has its execution cost associated with it. For example, the cost of the plan (134)(2) is computed as the sum of $Shared(G_E,\{1,3,4\})$ and $NonShared(G_E^i,2)$ (Equation~\ref{eq:dynamic-benefit}). 
The goal of the dynamic \app\ optimizer is to find a plan with minimal execution cost.

\begin{figure}[!htb]
\centering
\includegraphics[width=0.6\columnwidth]{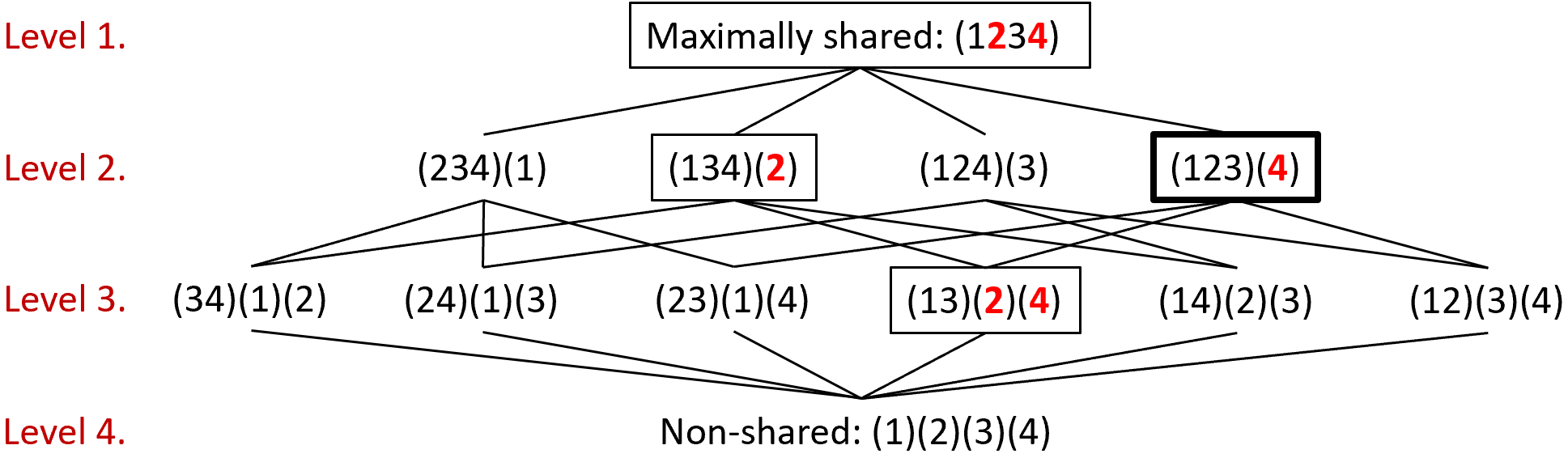}
\caption{Search space of sharing plans}
\label{fig:search-space}
\end{figure}

Traversing the exponential search space for each Kleene sub-pattern and each burst of events would jeopardize real-time responsiveness of \app. Fortunately, most plans in this search space can be pruned without loosing optimality  (Theorems~\ref{theo:pruning-1} and \ref{theo:pruning-2}). 
Intuitively, Theorem~\ref{theo:pruning-1} states that it is always beneficial to share the execution of a query that introduces no new snapshots.

\begin{theorem}
Let $E+$ be a Kleene sub-pattern that is shared by a set of queries $Q_E$ and not shared by a set of queries $Q_N$,
$Q_E \cap Q_N = \emptyset$,
$k_s = |Q_E|$, and
$k_n = |Q_N|$.
For a burst of events of type $E$,
let $q \in Q_E$ be a query that does not introduce new snapshots due to predicates for this burst of events (Definition~\ref{def:snapshot2}). 
Then the following follows:
\begin{align}
& \mathit{Shared}(Q_E) + \mathit{NonShared}(Q_N) \leq \nonumber\\
& \mathit{Shared}(Q_E \setminus \{q\}) + \mathit{NonShared}(Q_N \cup \{q\}) \nonumber
\end{align}
\label{theo:pruning-1}
\end{theorem}
\vspace{-5mm}

\begin{proof}
Equation~\ref{eq:one} summarizes the cost of sharing the execution of queries $Q_E$ where $q \in Q_E$. 
\begin{align}
&\mathit{Shared}(Q_E) + \mathit{NonShared}(Q_N) \nonumber\\
&= k_s \times \underline{s_c \times g \times p} \nonumber\\
&+ b \times (\log_2(g) + n \times s_p) \nonumber\\
&+ k_n \times b \times (\log_2(g) + n)
\label{eq:one}
\end{align}

Now assume the execution of $q$ is not shared with other queries in $Q_E$. That is, $q$ is removed from set $Q_E$ and added to set $Q_N$. Then, $k_s$ is decremented by one and $k_n$ is incremented by one in Equation~\ref{eq:two}. All other cost factors remain unchanged. In particular, the number of created $s_c$ and propagated $s_p$ snapshots do not change.

\begin{align}
&\mathit{Shared}(Q_E \setminus \{q\}) + \mathit{NonShared}(Q_N \cup \{q\}) \nonumber\\
&= (k_s-1) \times s_c \times g \times p \nonumber\\
&+ b \times (\log_2(g) + n \times s_p) \nonumber\\
&+ (k_n+1) \times \underline{b \times (\log_2(g) + n)} 
\label{eq:two}
\end{align}

Equations~\ref{eq:one} and \ref{eq:two} differ by one additive factor $(s_c \times g \times p)$ if $q$ is shared versus one additive factor $(b \times (\log_2(g) + n))$ if $q$ is not shared. These additive factors are underlined in Equations~\ref{eq:one} and \ref{eq:two}. Since $s_c \leq b$, $g \leq n$, and the number of predecessor types $p$ per type per query is negligible compared to other cost factors, we conclude that $(s_c \times g \times p) \leq (b \times (\log_2(g) + n))$, i.e., it is beneficial to share the execution of $q$ with other queries in $Q_E$.
\end{proof}

We formulate the following pruning principle per Theorem~\ref{theo:pruning-1}.

\textbf{\textit{Snapshot-Driven Pruning Principle}}. 
Plans at Level 2 of the search space that do not share queries that introduced no snapshots are pruned. All descendants of such plans are also pruned.

\begin{example}
In Figure~\ref{fig:search-space}, assume queries 1 and 3 introduced no snapshots, while queries 2 and 4 introduced snapshots.
Then, four plans are considered because they share queries 1 and 3 with other queries. These plans are highlighted by frames. The other eight plans are pruned since they are guaranteed to have higher execution costs.
\end{example}

Theorem~\ref{theo:pruning-2} below states that 
if it is beneficial to share the execution of a query $q$ with other queries $Q$, a plan that processes the query $q$ separately from other queries $Q_E \subseteq Q$ will have higher execution costs than a plan that shares $q$ with $Q_E$.
The reverse of the statement also holds. Namely, if it is not beneficial to share the execution of a query $q$ with other queries $Q$, a plan that shares the execution of $q$ with other queries $Q_E \subseteq Q$ will have higher execution costs than a plan that processes $q$ separately from $Q_E$.


\begin{theorem}
Let $E+$ be a Kleene sub-pattern that is shareable by a set of queries $Q$, 
$Q = Q_E \cup Q_N$, and 
$q \in Q_E$. Then:
\begin{align}
\text{If } 
& \mathit{Shared}(Q) \leq \mathit{Shared}(Q \setminus \{q\}) + \mathit{NonShared}(q) \text{,} 
\label{eq:three}\\
\text{then } 
& \mathit{Shared}(Q_E) + \mathit{NonShared}(Q_N) \leq
\mathit{Shared}(Q_E \setminus \{q\}) +
\mathit{NonShared}(Q_N \cup \{q\}) 
\label{eq:four}
\end{align}

This statement also holds if we replace all $\leq$ by $\geq$.
\label{theo:pruning-2}
\end{theorem}

\begin{proof}[Proof Sketch]
In Equation~\ref{eq:three}, if we do not share the execution of query $q$ with queries $Q$ and the execution costs increase, this means that the cost for re-computing $q$ is higher than the cost of maintenance of snapshots introduced by $q$ due to predicates.
Similarly in Equation~\ref{eq:four}, if we move the query $q$ from the set of queries $Q_E$ that share their execution to the set of queries $Q_N$ that are processed separately, the overhead of recomputing $q$ will dominate the overhead of snapshot maintenance due to $q$.
The reverse of Theorem~\ref{theo:pruning-2} can be proven analogously. 
\end{proof}

We formulate the following pruning principle per Theorem~\ref{theo:pruning-2}.

\textbf{\textit{Benefit-Driven Pruning Principle}}.
Plans at Level 2 of the search space that do not share a query that is beneficial to share are pruned.
Plans at Level 2 of the search space that share a query that is not beneficial to share are pruned.
All descendants of such plans are also pruned.

\begin{example}
In Figure~\ref{fig:search-space}, if it is beneficial to share query 2, then we can safely prune all plans that process query 2 separately. That is, the plan (134)(2) and all its descendants are pruned.
Similarly, if it is not beneficial to share query 4, we can safely exclude all plans that share query 4. That is, all siblings of (123)(4) and their descendants are pruned.
The plan (123)(4) is chosen (highlighted by a bold frame).
\label{ex:search_space}
\end{example}

\textbf{\textit{Consequence of Pruning Principles}}.
Based on all plans at Levels 1 and 2 of the search space, the optimizer classifies each query in the workload as either shared or non-shared. Thus, it chooses the optimal plan without considering plans at levels below 2.

\textbf{Complexity Analysis}.
Given a burst of new events, let $m$ be the number of queries that introduce new snapshots to share the processing of this burst of events. The number of plans at Levels 1 and 2 of the search space is $m+1$. 
Thus both time and space complexity of sharing plan selection is $O(m)$.




\begin{theorem}
Within one burst, \app\ has optimal time complexity.
\label{theo:optimality}
\end{theorem}

\begin{proof}
For a given burst of events, \app\ optimizer makes a decision to share or not to share depending on the sharing benefit in Section~\ref{sec:decisions}. 
If it is not beneficial to share, each query is processed separately and has optimal time complexity~\cite{PLRM18}.
If it is beneficial to share a pattern $E+$ by a set of queries $Q_E$, the time complexity is also optimal since it is optimal for one query $q \in Q_E$~\cite{PLRM18} and other queries in $Q_E$ are processed for free.
The set of queries $Q_E$ is chosen such that the benefit of sharing is maximal (Theorems~\ref{theo:pruning-1} and \ref{theo:pruning-2}).
\end{proof}

\textbf{Granularity of \app\ Sharing Decision}.
\app\ runtime sharing decisions are made per burst of events (Section~\ref{sec:decisions}). There can be several bursts per window (Definition~\ref{def:burst}). Within one burst, \app\ has optimal time complexity (Theorem~\ref{theo:optimality}).
According to the complexity analysis in Section~\ref{sec:decisions}, the choice of the query set has linear time complexity in the number of queries $m$ that introduce snapshots due to predicates. 
By Section~\ref{sec:queries}, the merge of graphlets has linear time complexity in the number of events $g$ per graphlet. 
\app\ would be optimal per window if it could make sharing decisions at the end of each window. However, waiting until all events per window arrive could introduce delays and jeopardise real-time responsiveness. Due to this low latency constraint, \app\ makes sharing decisions per burst, achieving significant performance gain over competitors (Section~\ref{sec:exp_results}).


\section{General Trend Aggregation Queries}
\label{sec:other}

While we focused on simpler queries so far, 
we now sketch how \app\ can be extended to support a broad class of trend aggregation queries as per Definition~\ref{def:query}.

\textbf{Disjunctive or Conjunctive Pattern}.
Let $P$ be a disjunctive or a conjunctive pattern and $P_1, P_2$ be its sub-patterns (Definition~\ref{def:pattern}). In contrast to event sequence and Kleene patterns,
$P$ does not impose a time order constraint upon trends matched by $P_1$ and $P_2$. 
Let $\mycount(P)$ denote the number of trends matched by $P$. $\mycount(P)$ can be computed based on $\mycount(P_1)$ and $\mycount(P_2)$ as defined below. The processing of $P_1$ and $P_2$ can be shared. 
Let $P_{1,2}$ be the pattern that detects trends matched by both $P_1$ and $P_2$. 
%
%
Let 
$C_{1,2} = \mycount(P_{1,2})$,
$C_1 = \mycount(P_1) - C_{1,2}$, and
$C_2 = \mycount(P_2) - C_{1,2}$.
$C_{1,2}$ is subtracted to avoid counting trends matched by $P_{1,2}$ twice.

Disjunctive pattern $(P_1 \vee P_2)$  matches a trend that is a match of $P_1$ or $P_2$. 
$\mycount(P_1 \vee P_2) = C_1 + C_2 + C_{1,2}$.

Conjunctive pattern $(P_1 \wedge P_2)$  matches a pair of trends $tr_1$ and $tr_2$ where $tr_1$ is a match of $P_1$ and $tr_2$ is a match of $P_2$. 
$\mycount(P_1 \wedge P_2) = 
C_1 * C_2 + 
C_1 * C_{1,2} +
C_2 * C_{1,2} +
\binom{C_{1,2}}{2}$
since each trend detected only by $P_1$ (not by $P_2$) is combined with each trend detected only by $P_2$ (not by $P_1$). In addition, each trend detected by $P_{1,2}$ is combined with each other trend detected only by $P_1$, only by $P_2$, or by $P_{1,2}$.


\textbf{Pattern with Negation} $\seq(P_1,\mynot\ N, P_2)$ is 
split into positive $\seq(P_1,P_2)$ and negative $N$ sub-patterns at compile time. At runtime, we maintain separate graphs for positive and negative sub-patterns. When a negative sub-pattern $N$ finds a match $e_n$, we disallow connections from matches of $P_1$ before $e_n$ to matches of $P_2$ after $e_n$. Aggregates are computed the same way~\cite{PLRM18}.

\textbf{Nested Kleene Pattern} $P=(\seq(P_1,P_2+))+$.
Loops exist at template level but not at graph level because previous events connect to new events in a graph but never the other way around due to temporal order constraints (compare Figures~\ref{fig:templates} and \ref{fig:snapshot}). 
The processing of $P$ and its sub-patters can be shared by several queries containing these patterns as illustrated by Example~\ref{ex:nested}. 

\begin{figure}[ht]
\centering
\includegraphics[width=0.25\columnwidth]{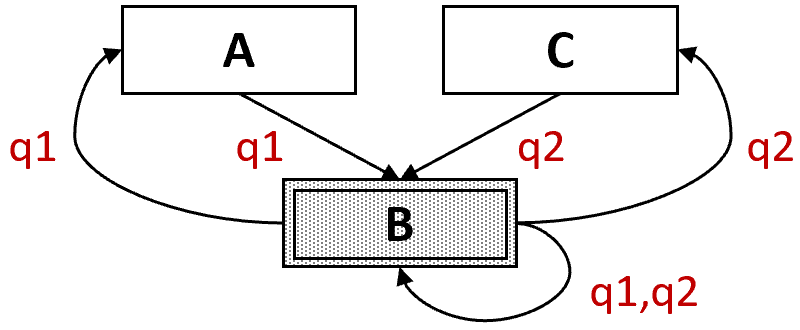}
\caption{\app\ query template for $Q$}
\label{fig:merged-template-2}
\end{figure}

\begin{example}
Consider query $q_1$ with pattern $(\seq(A,B+))+$ and query $q_2$ with pattern $(\seq(C,B+))+$. Figure~\ref{fig:merged-template-2} shows the merged template for the workload $Q=\{q_1,q_2\}$.
In contrast to the template in Figure~\ref{fig:merged_template}, there are two additional transitions (from $B$ to $A$ for $q_1$ and from $B$ to $C$ for $q_2$) forming two additional loops in the template.
Therefore, in addition to predecessor type relations in Example~\ref{ex:template_one_query}
($pt(B,q_1) = \{A,B\}$ and 
$pt(B,q_2) = \{C,B\}$),
two new predecessor type relations exist. 
Namely,
$pt(A,q_1) = \{B\}$ and 
$pt(C,q_2) = \{B\}$. 

Consider the stream in Figure~\ref{fig:snapshot}.
Similarly to Example~\ref{ex:three},
events in $A_1$--$C_2$ are predecessors of events in $B_3$, 
while events in $A_1$--$C_5$ are predecessors of events in $B_6$. 
Because of the additional predecessor type relations in the template in Figure~\ref{fig:merged-template-2}, 
events in $B_3$ are predecessors of events in $A_4$--$C_5$.
Due to these additional predecessor event relations,
more trends are now captured in the \app\ graph in Figure~\ref{fig:snapshot} and the intermediate and final trend counts have now higher values.
However, 
Definitions~\ref{def:graphlet}--\ref{def:dynamic-benefit} and Theorems~\ref{theo:pruning-1}--\ref{theo:optimality} hold and 
Algorithm~\ref{algo:snapshot-propagation} applies to share trend aggregation among queries in the workload $Q$.
\label{ex:nested}
\end{example}

%
%
%
%


%

\section{Experimental Evaluation}
\label{sec:experiments}

\subsection{Experimental Setup}

\textbf{Infrastructure}.
We have implemented \app\ in Java with JDK 1.8.0\_181 running on Ubuntu 14.04 with 16-core 3.4GHz CPU and 128GB of RAM. Our code is available online~\cite{hamlet_code}. We execute each experiment three times and report their average results here.

\textbf{Data Sets}. 
We evaluate \app\ using four data sets.

$\bullet$~\textit{New York city taxi and Uber real data set}~\cite{uber1} contains 2.63 billion taxi and Uber trips in New York City in 2014--2015. Each event carries a time stamp in seconds, driver and rider identifiers, pick-up and drop-off locations, number of passengers, and price. The average number of events per minute is 200.

$\bullet$~\textit{Smart home real data set}~\cite{smarthome} contains 4055 million measurements for 2125 plugs in 40 houses. Each event carries a timestamp in seconds, measurement, house identifiers, and voltage measurement value. The average number of events per minute is 20K.

$\bullet$~\textit{Stock real data set}~\cite{stockStream} contains up to 20 years of stock price history. Our sample data contains 2 million transaction records of 220 companies for 8 hours. Each event carries a time stamp in minutes, company identifier, price, and volume. The average number of events per minute is 4.5K.

$\bullet$~\textit{Ridesharing data set} was created by our stream generator to control the rate and distribution of events of different types in the stream. This stream contains events of 20 event types such as request, pickup, travel, dropoff, cancel, etc. Each event carries a time stamp in seconds, driver and rider ids, request type, district, duration, and price. The attribute values are randomly generated. The average number of events per minute is 10K.

\textbf{Event Trend Aggregation Queries}.
For each data set, we generated workloads similar to queries $q_1$--$q_3$ in Figure~\ref{fig:queries}. We experimented with the  two types of workloads described below.

$\bullet$~The first workload focuses on sharing Kleene closure because this is the most expensive operator in event trend aggregation queries (Definition~2.2). 
Further, the sharing of Kleene closure is a much overlooked topic in the literature; while the sharing of other query clauses (windows, grouping, predicates, and aggregation) has been well-studied in prior research and systems~\cite{AW04, GSCL12, KWF06, LMTPT05, LRGGWAM11}. Thus, queries in this workload are similar to Examples~\ref{ex:template_one_query}--\ref{ex:search_space}. Namely, they have different patterns but their sharable Kleene sub-pattern, window, groupby clause, predicates, and aggregates are the same. We evaluate this workload in Figures~\ref{fig:executor1}--\ref{fig:executor2}.

$\bullet$~The second workload is more diverse since the queries have sharable Kleene patterns of length ranging from 1 to 3, windows sizes ranging from 5 to 20 minutes, different aggregates (e.g., COUNT, AVG, MAX, etc.), as well as groupbys and predicates on a variety of event types. We evaluate this workload in Figures~\ref{fig:optimizer}--\ref{fig:opt-memory}.
%

The rate of events differs in different real data sets~\cite{uber1, smarthome, stockStream} that we used in our experiments. The window sizes are also different in the query workloads per data set. To make the results comparable across data sets, we vary the number of events per minute by a speed-up factor; which corresponds to the number of events per window divided by the window size in minutes.
The default number of events per minute per data set is included in the description of each data set.
Unless stated otherwise, 
the workload consists of 50 queries.
We vary  major cost factors (Definition~\ref{def:dynamic-benefit}), namely,
the number of events and
the number of queries.


\textbf{Methodology}. 
We experimentally compare \app\ to the following state-of-the-art approaches:

$\bullet$~\textit{MCEP}~\cite{KS19} is the most recently published state-of-the-art shared two-step approach. MCEP constructs all event trends prior to computing their aggregation. As shown in~\cite{KS19}, it shares event trend construction. It outperforms other shared two-step approaches SPASS~\cite{RLR16} and MOTTO~\cite{ZVDH17}.

$\bullet$~\textit{\sharon}~\cite{PRLRM18} is a shared approach that computes event sequence aggregation online. That is, it avoids sequence construction by incrementally maintaining a count for each pattern. \sharon\ does not support Kleene closure. To mimic Kleene queries, we flatten them as follows. For each Kleene pattern $E+$, we estimate the length $l$ of the longest match of $E+$ and specify a set of fixed-length sequence queries that cover all possible lengths up to $l$. 

$\bullet$~\textit{\greta}~\cite{PLRM18} supports Kleene closure and computes event trend aggregation online, i.e, without constructing all event trends. It achieves this online event trend aggregation by encoding all matched events and their adjacency relationships in a graph. However, \greta\ does not optimize for sharing a workload of queries. That is, each query is processed independently as described in Section~\ref{sec:non-shared-baseline}.

\textbf{Metrics}.
We measure \textit{latency} in seconds as the average time difference between the time point of the aggregation result output by a query in the workload and the arrival time of the latest event that contributed to this result.
\textit{Throughput} corresponds to the average number of events processed by all queries per second.
\textit{Peak memory consumption}, measured in bytes, corresponds to the maximal memory required to store snapshot expressions for \app, the current event trend for MCEP, aggregates for \sharon, and matched events for \app, MCEP, and \greta.


\subsection{Experimental Results}
\label{sec:exp_results}


\textbf{\app\ versus State-of-the-art Approaches}.
In Figures~\ref{fig:executor1} and \ref{fig:memory}, we measure all metrics of all approaches while varying the number of events per minute from 10K to 20K and the number of queries in the workload from 5 to 25. We intentionally selected this setting to ensure that the two-step approach MCEP, the non-shared approach \greta, and the fixed-length sequence aggregation approach \sharon\ terminate within hours.

\begin{figure}[!htb]
\centering
\subfigure[Latency vs $\#$events]{
\includegraphics[width=0.3\columnwidth]{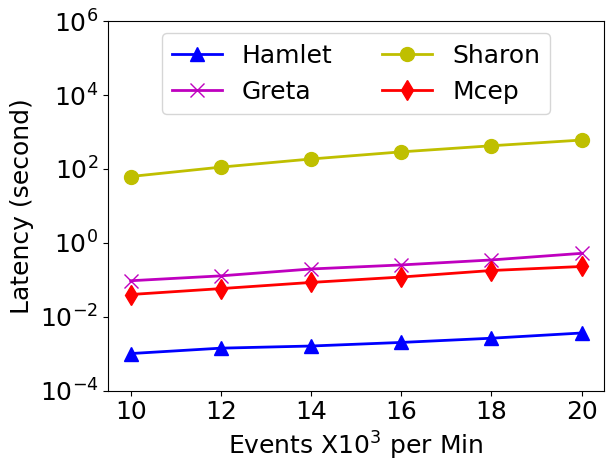}  
\label{fig:e1-latency-events}
}
\hspace{-10pt}
\subfigure[Latency vs $\#$queries]{
\includegraphics[width=0.3\columnwidth]{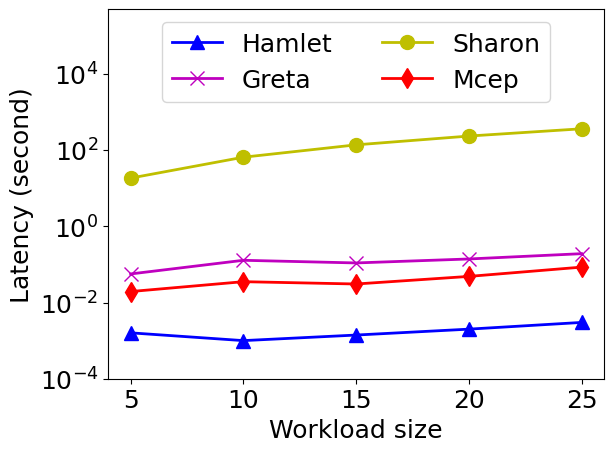}
\label{fig:e1-latency-queries}
}\\
\subfigure[Throughput vs $\#$events]{
\includegraphics[width=0.3\columnwidth]{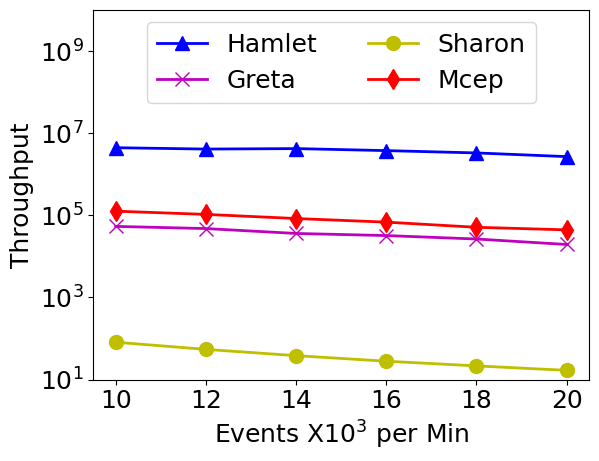}
\label{fig:e1-throughput-events}
}
\hspace{-10pt}
\subfigure[Throughput vs $\#$queries]{
\includegraphics[width=0.3\columnwidth]{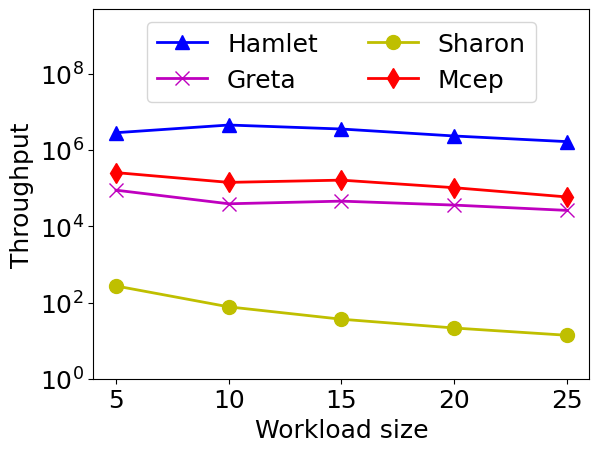}
\label{fig:e1-throughput-queries}
}
\caption{\app\ versus state-of-the-art approaches (Ridesharing)}
\label{fig:executor1}
\end{figure}

With respect to throughput, \app\ consistently outperforms 
\sharon\ by 3--5 orders of magnitude,
\greta\ by 1--2 orders of magnitude, and
MCEP 7--76-fold
(Figures~\ref{fig:e1-throughput-events} and \ref{fig:e1-throughput-queries}).
We observe similar improvement with respect to latency in Figures~\ref{fig:e1-latency-events} and \ref{fig:e1-latency-queries}. 
While \app\ terminates within 25 milliseconds in all cases, 
\sharon\ needs up to 50 minutes, 
\greta\ up to 3 seconds, and 
MCEP up to 1 second.
With respect to memory consumption, \app, \greta, and MCEP perform similarly, while \sharon\ requires 2--3 orders of magnitude more memory than \app\ in Figure~\ref{fig:memory}.

Such poor performance of \sharon\ is not surprising because \sharon\ does not natively support Kleene closure. To detect all Kleene matches, \sharon\ runs a workload of fixed-length sequence queries for each Kleene query. As Figure~\ref{fig:executor1} illustrates, this overhead dominates the latency and throughput of \sharon. 
In contrast to \sharon, \greta\ and MCEP terminate within a few seconds in this low setting because both approaches not only support Kleene closure but also optimize its processing. In particular, \greta\ computes trend aggregation without constructing the trends but does not share trend aggregation among different queries in the workload.
MCEP shares the construction of trends but computes trend aggregation as a post-processing step. Due to these limitations, \app\ outperforms both \greta\ and MCEP with respect to all metrics.

\begin{figure}[!htb]
	\centering
    \subfigure[Memory vs $\#$events]{
    	\includegraphics[width=0.3\columnwidth]{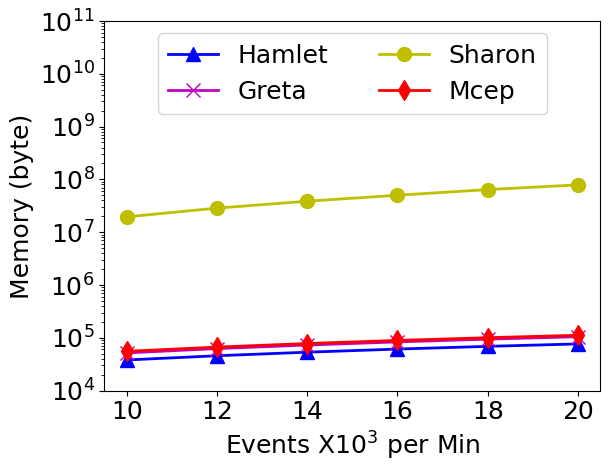}
    	 \label{fig:memory-queries}
	}
	\subfigure[Memory vs $\#$queries]{
    	\includegraphics[width=0.3\columnwidth]{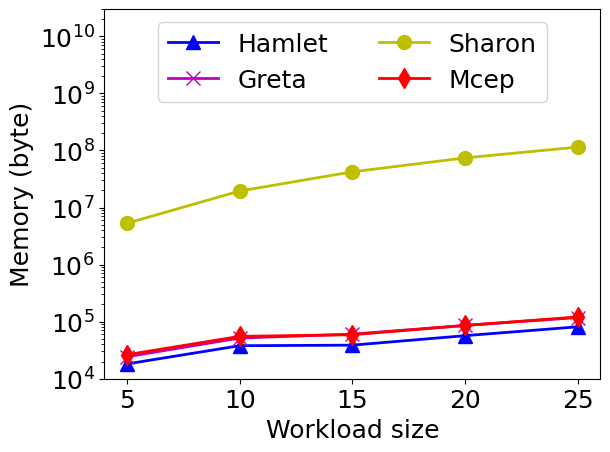}
    	\label{fig:memory-events}
	}
    \caption{\app\ vs state-of-the-art (Ridesharing)}
    \label{fig:memory}
\end{figure}

\begin{figure}[!htb]
	\centering
    \subfigure[Latency vs $\#$events (NYC)]{
    	\includegraphics[width=0.3\columnwidth]{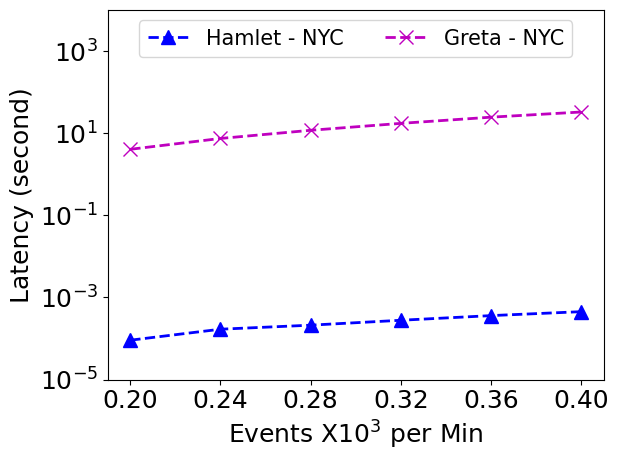}
    	 \label{fig:e2-latency-nyc-events}
	}\hspace{-8pt}
	\subfigure[Latency vs $\#$events (SH)]{
    	\includegraphics[width=0.3\columnwidth]{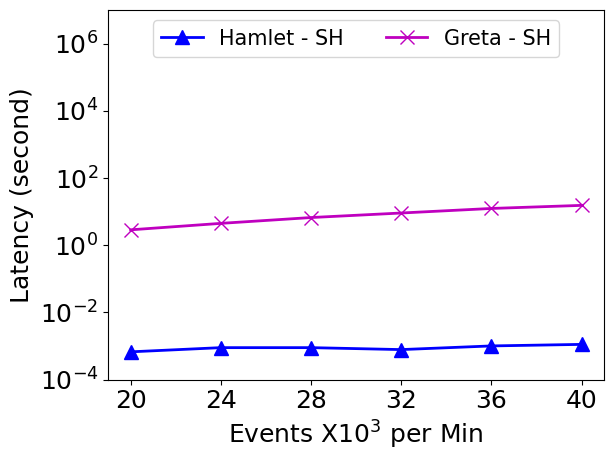}
    	 \label{fig:e2-latency-sh-events}
	}\\
	\subfigure[Throughput vs $\#$events (NYC)]{
    	\includegraphics[width=0.3\columnwidth]{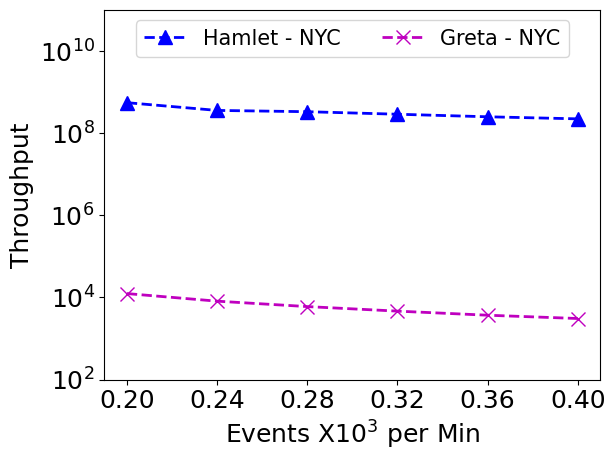}
    	 \label{fig:e2-throughput-nyc-events}
	}\hspace{-8pt}
	\subfigure[Throughput vs $\#$events (SH)]{
    	\includegraphics[width=0.3\columnwidth]{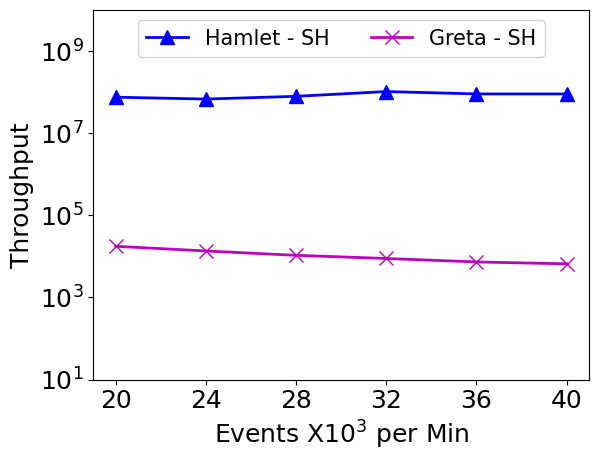}
    	\label{fig:e2-throughput-sh-events}
	}\\
	\subfigure[Memory vs $\#$events (NYC)]{
    	\includegraphics[width=0.3\columnwidth]{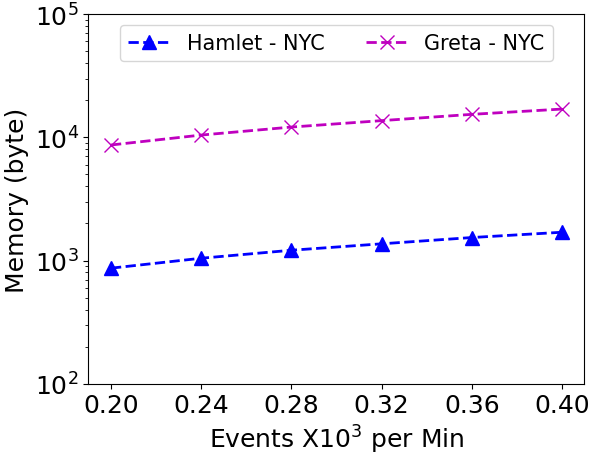}
    	 \label{fig:e2-memory-nyc-events}
	}\hspace{-8pt}
	\subfigure[Memory vs $\#$events (SH)]{
    	\includegraphics[width=0.3\columnwidth]{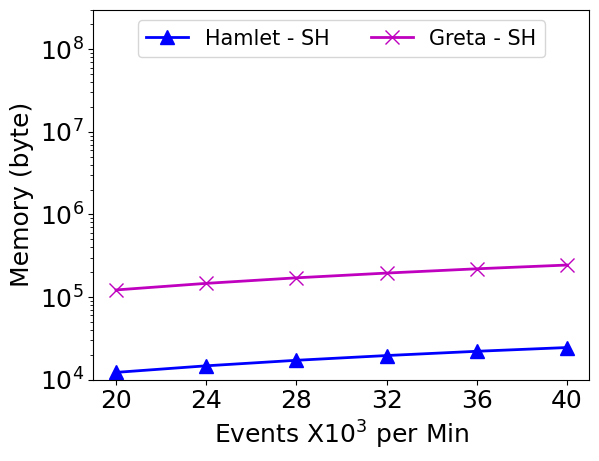}
    	\label{fig:e2-memory-sh-events}
	}\\
	\subfigure[Latency vs $\#$queries]{
    	\includegraphics[width=0.3\columnwidth]{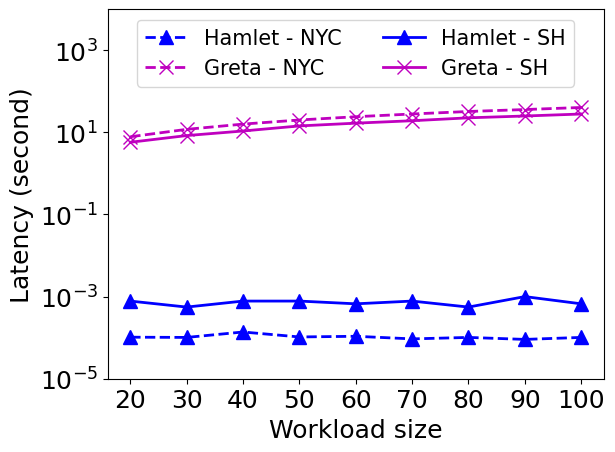}
    	 \label{fig:e2-latency-queries}
	}\hspace{-8pt}
	\subfigure[Throughput vs $\#$queries]{
    	\includegraphics[width=0.3\columnwidth]{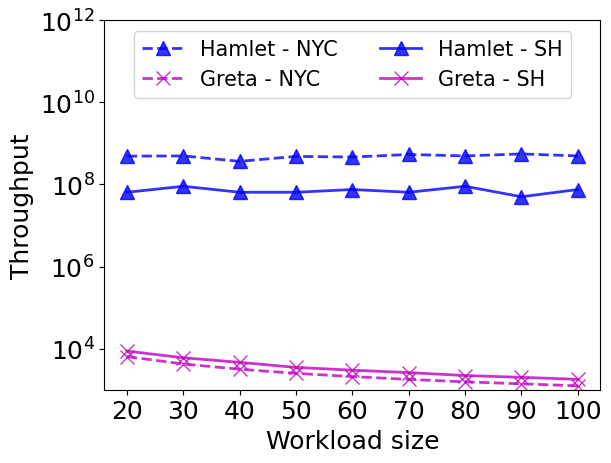}
    	\label{fig:e2-throughput-queries}
	}
    \caption{\app\ versus state-of-the-art approaches (NY City Taxi (NYC) and Smart Home (SH) data sets)}
    \label{fig:executor2}
\end{figure}

However, the low setting in Figures~\ref{fig:executor1} and \ref{fig:memory} does not reveal the full potential of \app. Thus in Figure~\ref{fig:executor2}, we compare \app\ to the most advanced state-of-the-art online trend aggregation approach \greta\ using two real data sets. We measure latency and throughput, while varying the number of events per minute and the number of queries in the workload. 
\app\ consistently outperforms \greta\ with respect to throughput and latency by 3--5 orders of magnitude. In practice this means that the response time of \app\ is within half a second, while \greta\ runs up to 2 hours and 17 minutes for 400 events per minute in Figure~\ref{fig:e2-latency-nyc-events}.

\textbf{Dynamic versus Static Sharing Decision}.
Figures~\ref{fig:optimizer} and \ref{fig:opt-memory} compare the effectiveness of \app\ dynamic sharing decisions to static sharing decisions. 
Each burst of events that can be shared contains 120 events on average in the stock data set. Our \app\ dynamic optimizer makes sharing decisions at runtime per each burst of events (Section~\ref{sec:dynamic-benefit}). The \app\ executor splits and merges graphlets at runtime based on these optimization instructions (Section~\ref{sec:decisions}). The number of all graphlets ranges from 400 to 600, while the number of shared graphlets ranges from 360 to 500. In this way, \app\ efficiently shares the beneficial Kleene sub-patterns within a subset of queries during its execution.

\begin{figure}[!htb]
	\centering
    \subfigure[Latency vs $\#$events]{
    	\includegraphics[width=0.3\columnwidth]{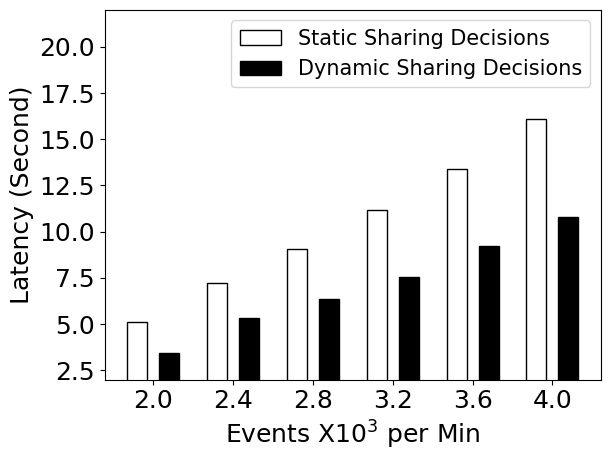}
    	 \label{fig:opt-latency-events}
	}\hspace{-8pt}
	\subfigure[Latency vs $\#$queries]{
    	\includegraphics[width=0.3\columnwidth]{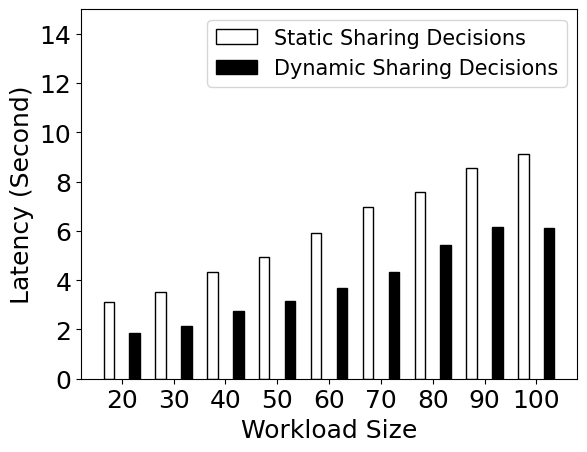}
    	 \label{fig:opt-latency-queries}
	}\\
	\subfigure[Throughput vs $\#$events]{
    	\includegraphics[width=0.3\columnwidth]{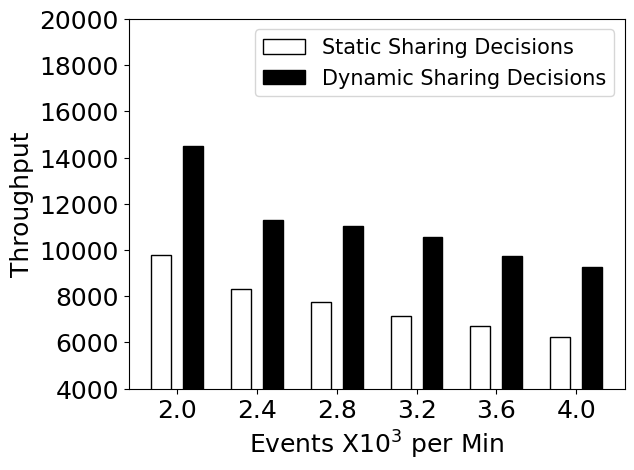}
    	 \label{fig:opt-throughput-events}
	}\hspace{-8pt}
	\subfigure[Throughput vs $\#$queries]{
    	\includegraphics[width=0.3\columnwidth]{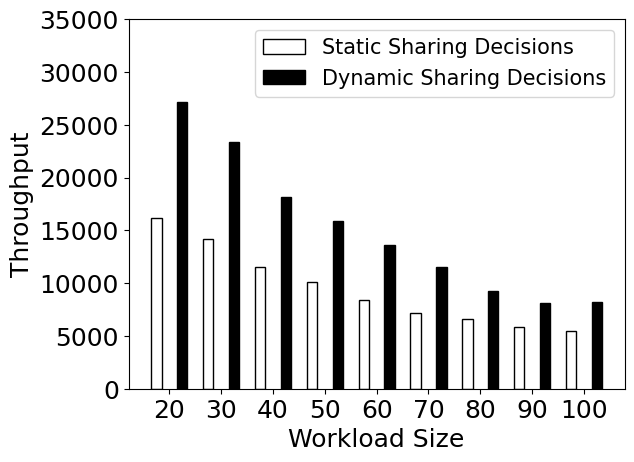}
    	\label{fig:opt-throughput-queries}
	}
    \caption{Dynamic versus static sharing decisions (Stock data set)}
    \label{fig:optimizer}
\end{figure}

In Figures~\ref{fig:opt-latency-events}, \ref{fig:opt-throughput-events} and \ref{fig:opt-memory-events}, as the number of events per minute increases from 2K to 4K, the number of snapshots maintained by the \app\ executor grows from 4K to 8K. As soon as the overhead of snapshot maintenance outweighs the benefits of sharing, the \app\ optimizer decides to stop sharing. The \app\ executor then splits these shared graphlets (Section~\ref{sec:decisions}). The \app\ dynamic optimizer shares approximately 90\% of bursts. The rest 10\% of the bursts are not shared which substantially reduces the number of snapshots by around 50\% compared to the shared execution. 

In contrast, the static optimizer decides to share certain Kleene sub-patterns by a fixed set of queries during the entire window. Since these decisions are made at compile time, they do not incur overhead at runtime. However, these static decisions do not take the stream fluctuations into account. Consequently, these sharing decisions may do more harm than good by introducing significant CPU overhead of snapshot maintenance, causing non-beneficial shared execution. During the entire execution, the static optimizer always decides to share, and the number of snapshots grows dramatically from 10K to 20K. 
Therefore, our \app\ dynamic sharing approach achieves 
21--34\% speed-up and
27--52\% throughput improvement
compared to the executor that obeys to static sharing decisions.


\begin{figure}[!htb]
\vspace{-10pt}
	\centering
    \subfigure[Memory vs $\#$events]{
    	\includegraphics[width=0.3\columnwidth]{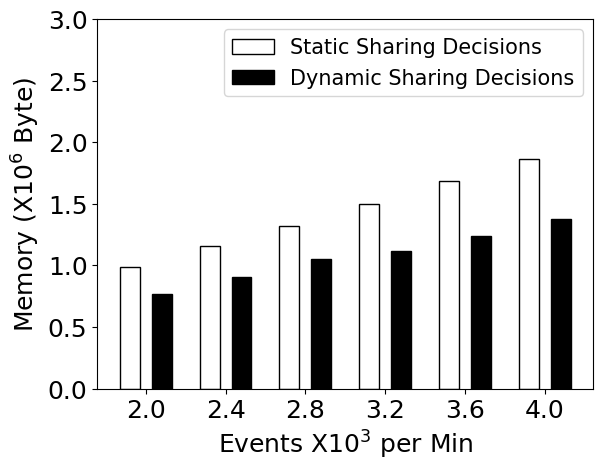}
    	 \label{fig:opt-memory-events}
	}
	\subfigure[Memory vs $\#$queries]{
    	\includegraphics[width=0.3\columnwidth]{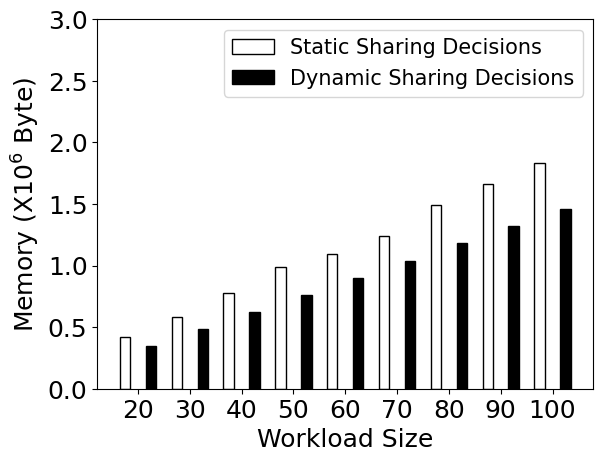}
    	\label{fig:opt-memory-queries}
	}
\vspace{-10pt}
    \caption{Dynamic versus static sharing decisions (Stock data set)}
    \label{fig:opt-memory}
\vspace{-5pt}
\end{figure}

We observe similar gains of \app\ with respect to memory consumption in Figure~\ref{fig:opt-memory-events}. \app\ reduces memory by 25\% compared to the executor based on static sharing decisions because the number of snapshots introduced by \app\ dynamic sharing decisions is much less than the number of snapshots introduced by the static sharing decisions. 


We also vary the number of queries in the workload from 20 to 100, and we observe similar gains by \app\ dynamic sharing optimizer in terms of latency, throughput, and memory (depicted in Figures~\ref{fig:opt-latency-queries}, \ref{fig:opt-throughput-queries}, and \ref{fig:opt-memory-queries}). \app\ can effectively leverage the beneficial sharing opportunities within a large query workload.



Lastly, we measured the runtime overhead of the \app\ dynamic sharing decisions. Even though the number of sharing decisions ranges between 400 and 600 per window, the latency incurred by these decisions stays within 20 milliseconds (less than 0.2\% of total latency per window) because these decisions are light-weight (Section~\ref{sec:decisions}). Also, the latency of one-time static workload analysis (Section~\ref{sec:tempalte}) stays within 81 milliseconds. Thus, we conclude that the overhead of dynamic decision making and static workload analysis are negligible compared to their gains.




\section{Related Work}
\label{sec:related}


\textbf{Complex Event Processing Systems} (CEP) have gained popularity in the recent years~\cite{esper,flink,streaminsight,oracle}. Some approaches use a Finite State Automaton (FSA) as an execution framework for pattern matching~\cite{ADGI08,DGPRSW07,WDR06,ZDI14}. Others employ tree-based models~\cite{MM09}. Some approaches study lazy match detection~\cite{KSS15}, compact event graph encoding~\cite{PLAR17}, and join plan generation~\cite{KS18join}. We refer to the recent survey~\cite{Giatrakos2020} for further details.
While these approaches support trend aggregation, they construct trends prior to their aggregation. Since the number of trends is exponential in the number of events per window~\cite{QCRR14, ZDI14}, such two-step approaches do not guarantee real-time response~\cite{PLRM18,PLRM19}. Worse yet, they do not leverage sharing opportunities in the workload. The re-computation overhead is substantial for workloads with thousands of queries.

\textbf{Online Event Trend Aggregation.} Similarly to single-event aggregation, event trend aggregation has been actively studied. A-Seq~\cite{QCRR14} introduces online aggregation of event sequences, i.e., sequence aggregation without sequence construction. \greta~\cite{PLRM18} extends A-Seq by Kleene closure. Cogra~\cite{PLRM19} further generalizes online trend aggregation by various event matching semantics. However, none of these approaches addresses the challenges of multi-query workloads, which is our focus.

\textbf{CEP Multi-query Optimization} follows the principles commonly used in relational database systems~\cite{Sellis:1988}, while focusing on pattern sharing techniques. RUMOR~\cite{hong2009rule} defines a set of rules for merging queries in NFA-based RDBMS and stream processing systems. E-Cube~\cite{LRGGWAM11} inserts sequence queries into a hierarchy based on concept and pattern refinement relations. SPASS~\cite{RLR16} estimates the benefit of sharing for event sequence construction using intra-query and inter-query event correlations. MOTTO~\cite{ZVDH17} applies merge, decomposition, and operator transformation techniques to re-write pattern matching queries. Kolchinsky et al.~\cite{KS19} combine sharing and pattern reordering optimizations for both NFA-based and tree-based query plans.
However, these approaches do not support online aggregation of event sequences, i.e., they construct all event sequences prior to their aggregation, which degrades query performance. To the best of our knowledge, \sharon~\cite{PRLRM18} and Muse~\cite{RPLR20} are the only solutions that support shared online aggregation. However, \sharon\ does not support Kleene closure. Worse yet, \sharon\ and Muse make static sharing decisions. In contrast, \app\ harnesses additional sharing benefit thanks to dynamic sharing decisions depending on the current stream properties.

\textbf{Multi-query Processing over Data Streams.} Sharing query processing techniques are well-studied for streaming systems. NiagaraCQ~\cite{10.1145/342009.335432} is a large-scale system for processing multiple continuous queries over streams. TelegraphCQ~\cite{DBLP:conf/cidr/ChandrasekaranDFHHKMRRS03} introduces a tuple-based dynamic routing for inter-query sharing~\cite{10.1145/564691.564698}. 
AStream~\cite{KRM19} shares computation and resources among several queries executed in Flink~\cite{flink}.
Several approaches focus on sharing optimizations given different predicates, grouping, or window clauses~\cite{AW04,GSCL12,HFAE03,KWF06,LMTPT05,TKPP20,ZKOSZ10}. However, these approaches evaluate Select-Project-Join queries with windows and aggregate single events. They do not support CEP-specific operators such as event sequence and Kleene closure that treat the order of events as a first-class citizen. Typically, they require the construction of join results prior to their aggregation. In contrast, \app\ not only avoids the expensive event trend construction, but also exploits the sharing opportunities among trend aggregation queries with diverse Kleene patterns.


\section{Conclusions}
\label{sec:conclusions}

\app\ integrates a shared online trend aggregation execution strategy with a dynamic sharing optimizer to maximize the benefit of sharing. 
It monitors fluctuating streams, recomputes the sharing benefit, and switches between shared and non-shared execution at runtime. 
Our experimental evaluation demonstrates substantial performance gains of \app\ compared to state-of-the-art.

\section*{Acknowledgments} 
This work was supported by 
NSF grants IIS-1815866, IIS-1018443, CRI-1305258, 
the U.S. Department of Agriculture grant 1023720,
and the U.S. Department of Education grant P200A150306.

\bibliographystyle{abbrv}
\bibliography{hamlet_tr}  



\end{document}